\documentclass[a4paper,10pt]{amsart}
\usepackage{amssymb, amsmath, graphicx, amsthm,amsfonts}

\usepackage{dsfont} 

\usepackage{array}
\newcolumntype{P}[1]{>{\centering\arraybackslash}p{#1}}

\usepackage{slashed} 

\usepackage{tikz-cd}

\usepackage[utf8]{inputenc}

\usepackage{hyperref} 

\theoremstyle{plain}

\theoremstyle{definition}
\theoremstyle{plain}
\newtheorem{theorem}{Theorem} 
\theoremstyle{plain}
\newtheorem{lemma}[theorem]{Lemma}
\newtheorem{proposition}[theorem]{Proposition}

\theoremstyle{definition}
\newtheorem{example}{Example}

\DeclareMathOperator{\diag}{diag}

\newcommand\summaryname{Abstract}
{\small\begin{center}%
		\bfseries{\summaryname} \end{center}}

\DeclareMathOperator{\tr}{Tr}

\title[Construction of Exact Solutions]{Construction of Exact Solutions to Nahm's Equations For the Multimonopole}
\author{H.W. Braden}

\address{
School of Mathematics and Maxwell Institute for Mathematical Sciences\\ The University of Edinburgh\\ 
Edinburgh EH9 3FD, Scotland, U.K.
}

\email{hwb@ed.ac.uk}

\author{Sergey A. Cherkis}
\address{
Department of Mathematics\\ University of Arizona\\
Tucson AZ 85721-0089, USA
}

\email{cherkis@math.arizona.edu}

\author{Jason M.~ Quinones}
\address{School of STAMP\\ Gallaudet University\\ Washington DC 20002, USA}
\email{jason.quinones@gallaudet.edu}

\begin{document}

\begin{abstract} We construct high rank solutions to Nahm's equations for boundary conditions that correspond to the Dirac multimonopole. Here, the spectral curve is explicitly known and we achieve the integration by constructing a basis of polynomial tuples  that forms a frame for the flow of the eigenline bundle over the curve.
\end{abstract}

\maketitle

\section{Introduction}
Nahm's equations, discovered in the context of monopoles in gauge theory \cite{nahm_simple_1980}, arise naturally in hyperk\"ahler geometry
and representation theory and have become an area of independent interest. These equations 
are
\begin{align} \label{eq:NahmT0}
\begin{split} \frac{dT_1}{ds}+[T_0,T_1]&=[T_2,T_3], \\
\frac{dT_2}{ds}+[T_0,T_2]&=[T_3,T_1], \\
\frac{dT_3}{ds}+[T_0,T_3]&=[T_1,T_2], \\
\end{split}
\end{align}
where $T_0(s), T_1(s), T_2(s), T_3(s)$ are antihermitian $n\times n$ matrix-valued functions and the real parameter $s$ lies in an interval, or union of intervals, depending on the problem. The boundary condition for the $T_*$'s
also depend critically on the problem. Nahm's equations form an integrable system, a system of nonlinear ordinary differential equations with sufficiently many conserved quantities to be solved by means of algebraic geometry,
and integrable systems techniques (such as the Baker-Akhiezer function) may be applied to their
solution.

In this paper we construct high rank solutions to Nahm's equations for the interval $(0,\infty)$ 
subject to the boundary conditions 
\begin{align}
\label{eq:NahmBddConditionNoT0}
\begin{split}
\underset{s \to 0}{\lim}& s T_0(s)=0, \quad
\lim_{s \to 0} sT_j(s)=\frac{i\sigma_j}{2},\ j=1,2,3,
\\
\lim_{s \to \infty}& (T_0(s),T_1(s),T_2(s),T_3(s)) \in \text{ad}_{U(n)} (0,i\tau_1,i\tau_2,i\tau_3).
\end{split}
\end{align}
Here, the triplet $\sigma_j$  is a specified $n$-dimensional irreducible representation of $\mathfrak{su}(2)$, $[\sigma_i,\sigma_j]=2i\epsilon_{ijk}\sigma_k$, and
by a constant gauge transformation we take this to be $\sigma_3=\text{diag}(n-1,n-3,\dots,-n+1)$ and $\sigma_+:=\frac{\sigma_1+i\sigma_2}{2}=[\sqrt{j(n-j)}\delta_{j,j-1}].$ The chosen triplet $(\tau_1,\tau_2,\tau_3)$ is regular in  $U(n)$; that  is the stabilizer of $\tau_1, \tau_2$, and $\tau_3$ is a maximal torus. This problem was related to the nonabelian Toda equation in \cite[\S4]{mikhaylov_12} where it was studied in terms of the 
Baker-Akhiezer functions.

The boundary conditions just described correspond to the Dirac $U(1)$ multimonopole \cite{nahm_abelian_1980}. 
As with a great number of other integrable systems, including the Korteweg-de Vries equation and the nonlinear Schr\"odinger equations, Nahm's equations arise from a symmetry reduction of the anti-self-dual equation (ASD)
\begin{align} \label{eq:ASDequation}
\ast F_A = - F_A,
\end{align}
for a principal $G$-bundle $E$ and connection $A$ over a 4-dimensional manifold.
The symmetry reduction of ASD \eqref{eq:ASDequation} to Nahm's \eqref{eq:NahmT0} is accomplished by taking $\Gamma=\mathbb{R}^3$ to be the closed subgroup of $\mathbb{R}^4$, the group of translations, and requiring the bundle and connection on $\mathbb{R}^4$ to be invariant under $\Gamma$. This gives us a (reduced) self-dual pair $(A,E)$ over $X=\mathbb{R}^4 / \Gamma$ that satisfies Nahm's equations. If one takes the dual space $X^\ast=\mathbb{R}^4 / \Gamma^\ast$, where $\Gamma^\ast$ is isomorphic to $\mathbb{R}$, the reduced self-dual pair $(\hat{A},\hat{E})$ is the monopole. This correspondence between equations on
$X$ and $X\sp\ast$ is known as the Nahm transform,
an operation here that transforms solutions of Nahm's equations into monopoles, and vice versa \cite{nahm_simple_1980} \cite{nahm_algebraic_1983}.
As we remarked earlier, the monopole that corresponds to the solution of Nahm's equations with the boundary conditions above is the Dirac $U(1)$ multimonopole \cite{nahm_abelian_1980}.
Much study has been devoted to the Nahm transform and the 4-manifolds $X$ for which this exists
\cite{corrigan_goddard_84, nakajima_93, jardim_04, charbonneau_hurtubise_19}. Recently, these acquired new significance as an essential ingredient of Witten's approach to the categorification of knot invariants \cite{witten_knots_12,witten_12b,gaiotto_witten_12,mazzeo_witten_14,mazzeo_witten_20,taubes_20}.

In this paper we solve Nahm's equations by constructing an orthonormal basis of polynomials that satisfy a matching condition, following the algebro-geometric method of integrating Nahm's equations. This method uses the Lax representation with spectral parameter $\zeta \in \mathbb{P}^1$ for matrix-valued functions $L$ and $M$ such that Nahm's equations are equivalent to
\begin{align}\label{lax}
\frac{d}{ds}L(s,\zeta) =[L(s,\zeta),M(s,\zeta)].
\end{align}
The Lax equation implies that the spectrum of $L(s,\zeta)$ is independent of the variable $s$ (as $\tr L^k$ are the conserved quantities of this system thanks to (\ref{lax})). This produces an algebraic curve $\mathcal{C}$ in $T\mathbb{P}^1$ defined by
\begin{align}
\det(\eta \mathds{1}-L(s,\zeta))=0,
\end{align}
for $\eta \frac{\partial}{\partial \zeta} \in T\mathbb{P}^1$. One then defines for generic $L(s,\zeta)$ a line bundle $L^s(n-1)$ over $\mathcal{C}$ with fiber at $(\zeta,\eta)$ given by the $\eta$-eigenspace of $L(s,\zeta)$. The flow of $L(s,\zeta)$ as $s$ evolves then produces a linear flow $L^s(n-1)$ in the Jacobian $\text{Jac}(\mathcal{C})$ of the spectral curve \cite{griffiths_linearizing_1985}, or, as we shall describe for this problem, the generalized Jacobian of $\mathcal{C}$. The integration concludes with writing a solution for Nahm's equations in terms of an appropriate frame of the linear flow. 

The task is then to construct an orthogonal frame of $H\sp0(L^s(n-1))$ for $s \in (0,\infty)$. For our boundary conditions we may express this frame in terms of polynomial tuples satisfying conditions we will prescribe; we give two equivalent systems of linear equations for the construction of these polynomials.
The polynomial tuples have significance beyond the construction of Nahm solutions. The Nahm transform between monopoles and Nahm solutions requires solving for normalizable zero modes to the associated Dirac operators. Once one obtains the orthogonal basis of polynomial tuples, it is a straightforward operation to construct all normalizable zero modes to the Nahm-sided Dirac operator following \cite{braden_construction_2018} and to the monopole-sided Dirac operator following \cite{lamy-poirier_dirac_2015}.

\section{Analytical Properties of Nahm Solutions} In this section, we discuss some analytical properties of a solution to Nahm's equation we will need.

First observe that 
Nahm's equations \eqref{eq:NahmT0} and boundary conditions \eqref{eq:NahmBddConditionNoT0} are invariant under
the gauge group $\mathcal{G}$,
\begin{align*}
	\mathcal{G}=\{g:[0,\infty) \to U(n) \mid g(0)=1,\ \underset{s \to \infty}{\lim}  g \in \text{Diag}(U(n)),\ g^{-1}(s)\dot{g}(s) \in L^2[0,\infty) \},
\end{align*}
where the action of $\mathcal{G}$ on $(T_0,T_1,T_2,T_3)$  is given by
\begin{align} 
T_0& \to gT_0g^{-1}-\dot{g}g^{-1}, \quad
T_j \to gT_jg^{-1}, \ \  j=1,2,3.
\end{align}
Using this gauge invariance we may transform $T_0$ to be identically zero ($T_0=0$) and deal with
the simpler form of Nahm's equations,
\begin{equation} \label{eq:NahmNoT0}
\frac{dT_1}{ds}=[T_2,T_3], \quad
\frac{dT_2}{ds}=[T_3,T_1], \quad
\frac{dT_3}{ds}=[T_1,T_2],
\end{equation}
under the same boundary conditions as \eqref{eq:NahmBddConditionNoT0}. 

Next, using the two lemmas from the literature below, we may
write the solution to \eqref{eq:NahmNoT0} satisfying
\eqref{eq:NahmBddConditionNoT0} in the form
\begin{align} \label{eq:AnalyticConditionsNahmData}
T_j(s)=\frac{i\sigma_j}{2s}+g_0i\tau_jg_0^{-1}+b_j(s),
\end{align}
where $g_0 \in U(n)$ and\footnote{We use the Frobenius norm $|b_j|=(\text{Tr}\bar{b}_j^T b_j)^{1/2}$, and we say $b_j(s) \in L^2[0,\infty)$ if $\int_0^\infty |b_j|^2 ds < \infty$.}
$b_j(s) \in L^2[0,\infty)$, 
$\underset{s \to \infty}{\lim} b_j(s)=0$. 
First, the behavior near the pole at $s=0$ is addressed by the following lemma.
\begin{lemma}[{\cite[Lemmas 5\&7]{cherkis_instantonsTaubNUT}}] Suppose the Nahm solution $(T_1(s),T_2(s),T_3(s))$ satisfies the following condition.
\begin{align} 
T_j - \frac{i \sigma_j}{2s} \in L^2([0,\epsilon)),
\end{align}
for some $\epsilon >0$.
Then
\begin{align}
T_j-\frac{i \sigma_j}{2s} \in L^\infty([0,\epsilon]).
\end{align}
\end{lemma}
Then, as $s \to \infty$, the solutions $(T_1,T_2,T_3)$ approach their limit exponentially fast, and the form
(\ref{eq:AnalyticConditionsNahmData}) follows.
The precise statement of the decay is given by the following lemma.
\begin{lemma}[{\cite[Lemma 3.4]{kronheimer_hyper-kahlerian_1990}}] \label{decay} Let $(T_1,T_2,T_3)$ satisfy Nahm's equations with the boundary condition \eqref{eq:NahmBddConditionNoT0} where $\tau_j$ is a regular triple and $\underset{s \to \infty}{\lim} T_j(s)=g_0i\tau_jg_0^{-1}$ for some $g_0 \in U(n)$. Then away from $s=0$,  there exists a constant $\kappa >0$ depending only on $(\tau_1,\tau_2,\tau_3)$ such that $|T_j - Ad(g_0)i\tau_j| \leq \text{const} \times e^{-\kappa s}$. \end{lemma} 

It is in the proof of this lemma 
where we see the requirement for $(\tau_1,\tau_2,\tau_3)$ to be a regular triplet. Here Nahm's equations \eqref{eq:NahmNoT0} are expressed as the gradient-flow equations for the function \begin{align*}\psi(T_1,T_2,T_3)=\tr T_1[T_2,T_3]. \end{align*} 
The critical set $C$ of this flow consists of triples $(T_1,T_2,T_3)$ which commute and the condition of regularity makes $C$ a smooth manifold in the neighborhood of $\text{Ad}(g_0)(i\tau_1,i\tau_2,i\tau_3)$. Exponential decay holds in general for any gradient system in the neighborhood of such a `hyperbolic' critical set with non-degenerate Hessian. (Any $\kappa$ of the lemma smaller than the smallest positive eigenvalue of the Hessian of $\psi$ will do.)
However, in the case that $(\tau_1,\tau_2,\tau_3)$ is not regular, the asymptotic behavior of Lemma \ref{decay} does not in general hold; simple counterexamples may be found in \cite[p.207]{kronheimer_hyper-kahlerian_1990}.

Further analysis then shows that with $T_0(s)=0$ there is an unique solution to the Nahm equations of \eqref{eq:NahmNoT0} with the required boundary conditions  \cite{atiyah_configurations_2001}.
The gauge transform of this solution by any $g\in \mathcal{G}$ thus produces a solution of \eqref{eq:NahmT0}
satisfying the boundary conditions \eqref{eq:NahmBddConditionNoT0} and so the solution
is unique only up to the action of the gauge group $\mathcal{G}$. 

In the next few sections, we construct a particular solution to \eqref{eq:NahmT0} satisfying \eqref{eq:NahmBddConditionNoT0}. We are not able to constructively solve for the gauge $g(s) \in \mathcal{G}$ that will put the Nahm solution in the form with $T_0(s)= 0$.

\section{The Spectral Approach}
\label{section:TheSpectralApproach}

The spectral method of integration begins with a Lax equation that is equivalent to the system of Nahm's equations. The mathematical formalism of this approach to integrating Nahm's equations was originally carried out by Hitchin in 1983 \cite{hitchin_construction_1983} for the $SU(2)$ monopole. We adopt the same approach.


In this section, we discuss the Lax  pair as well as the properties of the associated spectral curve $\mathcal{C}$ and eigenline bundle. We prove that for $s \in (0,\infty)$, the space of global sections of the line bundle $L^s(n-1)$ over $\mathcal{C}$ has dimension $n$.

\subsection{Lax Equation}

The Lax pair $(L,M)$ satisfying
\begin{align}
\frac{d}{ds}L = [L,M]
\end{align}
for Nahm's equations are given 
by
\begin{equation} \label{eq:NahmLaxN}
L^N=-i(T_1+iT_2)+2iT_3\zeta+i(T_1-iT_2)\zeta^2, \  M^N=T_0-i(T_3+(T_1-iT_2)\zeta), 
\end{equation}
{in the North patch $\mathbb{C}$ of
$\mathbb{P}\sp1=\mathbb{C}\cup\{\infty\}$, and by}
\begin{equation} \label{eq:NahmLaxS}
L^S=i(T_1-iT_2)+2iT_3\frac{1}{\zeta}-i(T_1+iT_2)\frac{1}{\zeta^2}, \ \  M^S=T_0+i(T_3-(T_1+iT_2)\frac{1}{\zeta}),
\end{equation}
{in the South patch $\{\zeta\ne0\}$ of
$\mathbb{P}\sp1$.}

Crucially, $L$ and $M$ have transition functions given by
\begin{align} \label{eq:transition function Nahm Lax}
L^S=\frac{1}{\zeta^2}L^N, \ \ \ M^S=M^N+\frac{1}{\zeta}L^N,
\end{align}
and $L$ and $M$ satisfy the reality relationship
\begin{align} \label{eq: reality Nahm L M}
L^N(-1/\bar{\zeta})^\dagger = -L^S(\zeta), \ \ \ \ \ M^N(-1/\bar{\zeta})^\dagger=-M^S(\zeta).
\end{align}

Conversely, if a Lax pair $(L,M)$ is holomorphic with respect to the spectral parameter $\zeta\in\mathbb{P}^1$ and satisfies the transition functions \eqref{eq:transition function Nahm Lax} as well as the reality conditions \eqref{eq: reality Nahm L M}, then it follows that $(L,M)$ has the form (\ref{eq:NahmLaxN},  \ref{eq:NahmLaxS}) with 
$(T_0,T_1,T_2,T_3)$
solving Nahm's equations. 

\subsection{Spectral Curve}
The principal feature of a Lax pair $(L,M)$ is that the spectrum of $L(s,\zeta)$ does not evolve in $s$. The transition function \eqref{eq:transition function Nahm Lax} for $L(s,\zeta)$ identifies its spectrum as a subset of the total space of the bundle $\mathcal{O}_{\mathbb{P}^1}(2)$. As $\mathcal{O}_{\mathbb{P}^1}(2) \cong T\mathbb{P}^1$, we may write the spectral curve $\mathcal{C}$ as the subset of $T\mathbb{P}^1$ satisfying
\begin{align}
\det (\eta - L(s,\zeta))=0,
\end{align}
for $\eta \frac{\partial}{\partial \zeta} \in T\mathbb{P}^1$.
Our boundary conditions allow us to write this spectral curve explicitly.

Let $(i\tau^1,i\tau^2,i\tau^3)$ be the regular triplet of \eqref{eq:NahmBddConditionNoT0}; without loss of generality
these commuting matrices may be taken to be diagonal. The diagonal entries of the triplet then pick out a configuration of $n$ points in $\mathbb{R}^3$ given\footnote{Here $(x_j^1,x_j^2,x_j^3)\in\mathbb{R}^3$ corresponds to the position of the $j$-th Dirac monopole.} by $(x_j^1,x_j^2,x_j^3)=\left(\tau_{jj}^1,\tau_{jj}^2,\tau_{jj}^3\right)$ for $j=1,2,\dots,n$.
We have
\begin{align}
\underset{s \to \infty}{\lim} L^N(s,\zeta) \in ad_{U(n)}\left(\begin{smallmatrix} p_1(\zeta) & 0 & \dots & 0 \\ 0 & p_2(\zeta) & \dots & 0 \\  & & \ddots &    \\ 0 & 0 & \dots & p_n(\zeta) \end{smallmatrix}\right)
\end{align}
where $p_{j}(\zeta)=x_j^1+ix_j^2 - 2x_j^3 \zeta - (x_j^1-ix_j^2)\zeta^2:=z_j - 2x_j^3 \zeta -\bar z_j \zeta^2 $. 

The spectral curve does not evolve in $s$ so it follows that $\mathcal{C} \subset T\mathbb{P}^1$ is \begin{align} \label{eq:Spectral Curve Equation}
 \mathcal{C}=\left\{ (\zeta,\eta) \in T\mathbb{P}^1 : \overset{n}{\underset{j=1}{\prod}} (\eta-p_j(\zeta))=0 \right\}.
	\end{align}
Therefore $\mathcal{C}$ is a degenerate $n$-sheeted covering of $\mathbb{P}^1$ where each sheet $\mathcal{C}_i=\{(\zeta,p_i(\zeta)) \in T\mathbb{P}^1\}$ is isomorphic to $\mathbb{P}^1$. We note that Hitchin's construction allows the spectral curve to be 
singular or reducible. The spectral curve is invariant under the antiholomorphic involution $\mathcal{J}(\zeta,\eta)
=(-1/\bar\zeta, -\bar\eta/\bar\zeta^2)$.

The double points correspond to values of $\zeta$ where $p_i(\zeta)=p_j(\zeta)$ (for $i\ne j$).
For two distinct points $x_i$ and $x_j$ of the configuration with distance $r_{ij}$, the sheets $\mathcal{C}_i$ and $\mathcal{C}_j$ intersect at the double points $\zeta=a_{ij}$ and $\zeta=a_{ji}$ where
\begin{align} \label{eq:Spectraldoublepoint}
a_{ij}=\frac{x_i^3-x_j^3+r_{ij}}{x_j^1-x_i^1-i(x_j^2-x_i^2)}.
\end{align}
Note that $a_{ij}=-1/\bar{a}_{ji}$.
The point $a_{ij}$ corresponds to the direction in $\mathbb{R}\sp3$ from $x_i$ to $x_j$ under the stereographic projection $(z,x^3) \mapsto \frac{z}{1+x^3}$. To see this, the direction from $x_i$ to $x_j$ is the unit vector $\frac{x_j-x_i}{r_{ij}}$. Under the stereographic projection, this becomes $\frac{z_j-z_i}{x_j^3-x_i^3+r_{ij}}=\frac{\left(z_j-z_i\right)\left(x_i^3-x_j^3+r_{ij}\right)}{|z_i-z_j|^2}$, which simplifies to the formula \eqref{eq:Spectraldoublepoint}.
Note, if $x_i$ and $x_j$ are vertically separated (i.e. $x_i^1=x_j^1$ and $x_i^2=x_j^2$) and, say $x_i^3 >x_j^3$, then $a_{ij}=\infty \in \mathbb{P}^1$ and $a_{ji}=0 \in \mathbb{P}^1$. For simplicity we will assume throughout that the $n$ Dirac monopoles are generically situated; 
other configurations may be obtained by the appropriate limit of this. 

\subsection{Eigenline Bundle}
The Nahm solutions produce a linear flow $L^s(n-1)$ in $\text{Jac}(\mathcal{C})$, where the superscript $s$ runs over the interval $(0,\infty)$. The line bundle $L^s(n-1)$ is the eigenline bundle over the spectral curve $\mathcal{C}$ of the Lax pair $(L,M)$ at time $s$. Let us describe how this eigenline bundle arises; because we are dealing with a singular
curve the line bundles are described by points in the generalized Jacobian.

First, we consider the solutions $U_j=(U_j^N,U_j^S)$ to the Lax linear problem associated to $(L,M)$, 
\begin{align} \label{eq:LaxLinearProblem} 
\begin{cases}
(\frac{d}{ds}+M^N)U_j^N=0, \\
(L^N-p_j(\zeta))U_j^N=0,
\end{cases} 
&&\begin{cases}
(\frac{d}{ds}+M^S)U_j^S=0, \\
(L^S-\frac{p_j(\zeta)}{\zeta^2})U_j^S=0.
\end{cases}
\end{align}
In other words, each $U_j$ is an eigenvector of $L$ for the eigenvalue $\eta=p_j(\zeta)$ on the sheet $\mathcal{C}_j$ that evolves in $s$ as $\frac{dU_j}{ds}=-M U_j$. The eigenline bundle $L^s(n-1)$ over $\mathcal{C}$ is the line bundle where the fiber over the sheet $\mathcal{C}_j$ is the eigenline spanned by $U_j$. 

Now the line bundle $L^s(n-1)$ is specified by the transition functions between the North and South patches and the behaviour at the double points of $\mathcal{C}$. They are determined by the transition functions \eqref{eq:transition function Nahm Lax} of $(L,M)$ and by the boundary condition \eqref{eq:NahmBddConditionNoT0} at $s=0$.

\begin{proposition}[\cite{braden_construction_2018}] \label{defn:EigenlineBundle} The eigenline bundle $L^s(n-1)$ for $s \in (0,\infty)$ over the spectral curve $\mathcal{C}$  with section $u(s,\zeta)$ has the transition functions 
	\begin{align} \label{NSTF}
	u^N(s,\zeta)=\zeta^{n-1} u^S(s,\zeta)\diag(e^{sp_1(\zeta)/\zeta},\dots,e^{sp_n(\zeta)/\zeta}),
	\end{align}
	and $u_i(a_{ij})=u_j(a_{ij})$ at the double points $a_{ij}$ for  $i,j=1,2,\dots,n$ with $i \neq j$, where $u_i$ and $u_j$ are the values of the section $u$ on the $\mathcal{C}_i$ and $\mathcal{C}_j$ sheets.
\end{proposition}
\begin{proof} Set $U=(U_1,\dots,U_n)$ to be the solutions to the Lax linear problem \eqref{eq:LaxLinearProblem}. We are looking for the transition function $F(s,\zeta)$ of $U^N=U^SF(s,\zeta)$.
From the transition of $L$ \eqref{eq:transition function Nahm Lax}, we have $L^N=\zeta^2 L^S$ so that $L^N$ and $L^S$ have the same eigenspace and therefore $U^N=U^SF(s,\zeta)$ for a diagonal matrix $F(s,\zeta)$.
From the transitions of $M$ \eqref{eq:transition function Nahm Lax},
\begin{align*}
0 &=\left(\frac{d}{ds}+M^N\right)U^N 
=\left(\frac{d}{ds}+M^S-\frac{1}{\zeta}L^N\right)U^SF \\ 
&=\left[\left(\frac{d}{ds}+M^S\right)U^S\right]F+U^S \frac{dF}{ds}-U^S\frac{\eta}{\zeta}F\\
&=U^S\left(\frac{d}{ds}-\frac{\eta}{\zeta}\right)F.
\end{align*}
Thus $F(\zeta,s)=g(\zeta)e^{s \frac{\eta}{\zeta}}$ for some $s$-independent diagonal matrix $g(\zeta)$. We will see later from Proposition \ref{prop: dimension of sections} that $g(\zeta)=\zeta^{n-1}$ and
so we have (\ref{NSTF}).

We now address what occurs to the eigenbundle $L^s(n-1)$ at the double points $\zeta=a_{ij}$ of \eqref{eq:Spectraldoublepoint} for the spectral curve $\mathcal{C}$. First we simplify our linear problem.
Setting $\sigma_\pm = \frac{1}{2}(\sigma_1+i\sigma_2)$, we have near $s=0$ that 
$
    T_j = \frac{i\sigma_j}{2s}+O(s^0)
$
(for $j=1,2,3$) and so 
\begin{equation*}
L^N(s):=-i(T_1+iT_2)+2iT_3\zeta+i(T_1-iT_2)\zeta^2 
=(\sigma_+-\zeta\sigma_3-\zeta^2\sigma_-)\frac{1}{s}+O(s^0).
\end{equation*}
Applying the relations
\begin{align} \label{eq:PauliExpRelation}
e^{\zeta \sigma_-}\sigma_3e^{-\zeta \sigma_-}=\sigma_3+2\zeta \sigma_-, \qquad e^{\zeta \sigma_-}\sigma_+e^{-\zeta \sigma_-}=\sigma_+-\zeta\sigma_3-\zeta^2\sigma_-,
\end{align}
it follows that (in the gauge $T_0=0$)
\begin{align}\label{gaugeLM}
L^N(s)=e^{\zeta \sigma_-}\big(\frac{\sigma_+}{s}+O(s^0)\big)e^{-\zeta \sigma_-}, \ \ \ \  \frac{d}{ds}+M^N=e^{\zeta \sigma_-}\big(\frac{d}{ds}+\frac{\sigma_3}{2s}+O(s^0)\big)e^{-\zeta \sigma_-}.
\end{align}
Next we may use a complex gauge transformation that is identity at $s=0$ to gauge $M^N$ to have no $O(s^0)$ term. The remaining complex gauge transformations that preserve this form of $M^N$ are $g'=e^{\zeta\sigma_-}s^{-\sigma_3/2}g_c s^{\sigma_3/2}e^{-\zeta\sigma_3}$ with any constant $g_c \in GL(n,\mathbb{C})$.

Since $L^N$ satisfies the Lax equation and has the form (\ref{gaugeLM}), its adjoint weight decomposition contains only negative weights.
Using the remaining complex gauge transformations $g'$ we may then put $L^N$
into the form\footnote{The observation regarding the weight decomposition means that 
(\ref{gaugeLM}) takes the form
$
L^N(s)=e^{\zeta \sigma_-}\Bigg[
\frac1{s}\sum_{i=1}\sp{n-1}E_{i\, i+1}
+ \sum_{i\ge j}a_{ij} E_{ij}\, s^{i-j}
\Bigg]e^{-\zeta \sigma_-}
$. Now take $g_c=\begin{pmatrix}
1&0&0&\ldots\\ g_{21}&1&0&\ldots\\
g_{31}&g_{32}&1&\ddots
\end{pmatrix}$. Then (ignoring the exponential factors) $g'_{ij}= g_{ij}s^{i-j}$ and similarly for the inverse $g'\sp{-1}$. 
Now we may solve for the main diagonal of
$g'\sp{-1} L^N g'$ to be the constant ($a_0$, say)
in terms of 
$g_{21}$,\ldots,$g_{n\, n-1}$ and
$a_0$; 
then we may recursively solve for the 
next lower diagonal to be constant ($a_1$, say) in terms of $g_{31}$,\ldots, $g_{n\, n-2}$
and $a_1$; and so on.
}
\begin{align}
  L^N(s)&=e^{\zeta \sigma_-}\left(\frac{\sigma_+}{s}+a_0+a_1s\sigma_-+a_2s^2\sigma_-^2+\dots+a_{n-1}s^{n-1}\sigma_-^{n-1}\right)e^{-\zeta \sigma_-} 
  \nonumber \\
&=e^{\zeta \sigma_-} \begin{pmatrix}
a_0&s^{-1}&0&0&\ldots&0\\
a_1 s&a_0&s^{-1}&0&\ldots&0\\
a_2 s^2&a_1 s&a_0&s^{-1}&&0\\
\vdots&\ddots&\ddots&\ddots&&\\
a_{n-2}s^{n-2}&&&&&s^{-1}\\
a_{n-1}s^{n-1}&a_{n-2}s^{n-2}&& &a_1 s&a_0
\end{pmatrix} e^{-\zeta \sigma_-}.
\end{align}
The coefficients $a_j=a_j(\zeta)$ are determined by the spectral curve. At this stage we have sufficiently simplifed $L^N$ and
$$
\frac{d}{ds}+M^N=e^{\zeta \sigma_-}\Bigg[\frac{d}{ds}+\frac{\sigma_3}{2s}\Bigg]e^{-\zeta \sigma_-}
$$
to solve the linear problem. Clearly $U=e^{\zeta \sigma_-}
s^{-\sigma_3/2} U'$ for any constant $U'$ satisfies $(\frac{d}{ds}+M^N)U=0$. Next observe that
\begin{equation*}
\begin{split}
s^{\sigma_3/2}&
\left(\frac{\sigma_+}{s}+a_0+a_1s\sigma_-+a_2s^2\sigma_-^2+\dots+a_{n-1}s^{n-1}\sigma_-^{n-1}\right)
s^{-\sigma_3/2}
\\ &=
\left(\sigma_+ +a_0+a_1\sigma_-+a_2\sigma_-^2+\dots+a_{n-1}\sigma_-^{n-1}\right).
\end{split}
\end{equation*}
Thus
$$
0=(L^N(s)-\eta)U=
e^{\zeta \sigma_-}
s^{-\sigma_3/2}\left[
\sigma_+ +a_0+a_1\sigma_-+a_2\sigma_-^2+\dots+a_{n-1}\sigma_-^{n-1} -\eta\right] U'
$$
and we have reduced the problem to the eigenvectors of a
constant matrix. If $R\sp{-1}
\left[
\sigma_+ +a_0+a_1\sigma_-+a_2\sigma_-^2+\dots+a_{n-1}\sigma_-^{n-1}\right]R$ is in rational canonical form, and so with eigenvectors proportional to $( 1 , \eta, \eta^2 ,\ldots , \eta\sp{n-1})\sp{T}$ for the eigenvalue $\eta$, we see we may write
\begin{align}
{U}'=R(\zeta)\begin{pmatrix}1 & 1 & \dots & 1 \\ p_1(\zeta) & p_2(\zeta) & \dots & p_n(\zeta) \\ \vdots & & & \vdots \\ p_1(\zeta)^{n-1} & p_2(\zeta)^{n-1} & \dots & p_n(\zeta)^{n-1} \end{pmatrix}\Lambda
\end{align}
for some diagonal matrix $\Lambda$.
From this we see that the eigenspaces of ${L^N}$ remain one-dimensional
at $\zeta=a_{ij}$ and we may set $\Lambda$ to be  the identity matrix. It follows that
\begin{align} \label{eq:SectionMatchingConditions}U_i(a_{ij})=U_j(a_{ij}).
\end{align}
Our analysis has thus far been for $s$ in the vicinity of $0$, but the results of the last section show that the eigenspace of $L(s,\zeta)$ remains one-dimensional for all finite $s$.

Each row of the matrix $U$ gives us a section $u$ of the line bundle $L^s(n-1)$. Crucially, there is an isomorphism between the sections of $L^s(n-1)$ and the sections $q$ of the line bundle defined by the transition functions
\begin{align} \label{eq:PolynomialRowsBundle}
\begin{split}
    q_i(a_{ij})&=e^{-sr_{ij}}q_j(a_{ij})   \text{ for } a_{ij} \text{ a double point of } \mathcal{C}, \\
    q^N(\zeta)&=\zeta^{n-1}q^S(\zeta).
    \end{split}
\end{align}
The isomorphism is obtained by using the splitting $\frac{p_j(\zeta)}{\zeta}=-h_j^+(\zeta)-h_j^-(\zeta)$, where $h_j^+(\zeta)=p_j^3+(p_j^1-ip_j^2)\zeta$ and $h_j^-(\zeta)=-\frac{p_j^1+ip_j^2}{\zeta}+p_j^3=\overline{h_j^+(-1/\overline{\zeta}) } $,
and then setting
\begin{align} \label{eq:QandURelationship}
	q^N(\zeta)=u^N(\zeta)\,\text{diag}(e^{sh_j^+(\zeta)}), \ \ \ q^S(\zeta)=u^S(\zeta)\,\text{diag}(e^{-sh_j^-(1/\zeta)}).
\end{align}
Here we have used that $h_i^+(a_{ij})- h_j^+(a_{ij})=-r_{ij}$.

The transition functions for $q$ imply that the restriction $q^N \lvert_{\mathcal{C}_i}(\zeta)$ to the $i$-th sheet is a polynomial in $\zeta$ of degree at most $n-1$.  By writing $q^N=(q^N\lvert_{\mathcal{C}_1},\dots,q^N\lvert_{\mathcal{C}_n})$, we see $q^N$ is a row of polynomials in $\zeta$ each of degree less than or equal to $n-1$  that satisfy the matching conditions $q \lvert_{\mathcal{C}_i}(a_{ij})=e^{-sr_{ij}} q \lvert_{\mathcal{C}_j}(a_{ij})$. For ease of notation we will write $Q_j(\zeta):=q^N \lvert_{\mathcal{C}_j}.$

\end{proof}

We now show that the space of sections of $L^s(n-1)$ has dimension $n$ for $s \in (0,\infty)$; we will need this fact to prove invertibility of the linear systems we use to construct the polynomials in later sections.

\begin{proposition} \label{prop: dimension of sections}
	For $s \in (0,\infty)$, the dimension 
	of $H^0(\mathcal{C},L^s(n-1))$ is $n$.
\end{proposition}
\begin{proof}
A direct application of the Riemann-Roch formula \cite[pp.165-166]{hitchin_construction_1983} implies that if $\dim H^0(\mathcal{C},L^s(n-2))=0$, then $\dim H^0(\mathcal{C},L^s(n-1))=n$. Therefore, in this proof, we consider the line bundle $L^s(n-2)$ and show the dimension of its global sections is zero. 

A section of $L^s(n-2)$ corresponds to an $n$-tuple of degree $n-2$ polynomials satisfying matching conditions 
	\begin{align} \label{eq:matching} P_i(s,a_{ij})=e^{-sr_{ij}}P_j(s,a_{ij}), \end{align} 
in the same vein as \eqref{eq:PolynomialRowsBundle}. 
The matching conditions \eqref{eq:matching} form a linear system for the unknown polynomial coefficients. The $s$ dependence enters the linear system via functions $e^{-sr_{ij}}$. Upon writing this linear system in the form $Ax=0$ (such a matrix $A$ is given in Subsection \ref{subsection: Lagrangian Approach}) the $s$-dependent terms in $A$ are in the form $e^{-s \Delta}$ for $\Delta>0$ and it is known that $L^s(n-2)$ has no nonzero sections when $s$ is large \cite[Cor. 3.3]{bielawski_reducible_2007}. 

The null space of $A$ may then be written as a series in $e^{-s \Delta}$ for various positive $\Delta$, 
and we then write the polynomials as a formal series, which we denote as its perturbation expansion,
\begin{align}
	P_i(s,\zeta)=g_i(\zeta)+e^{-s\Delta'_i}g_i'(\zeta)+e^{-s\Delta_i''}g''_i(\zeta)+\dots
\end{align}
	with $0 < \Delta_i' <\Delta_i''<\dots$ and $s$-independent polynomials $g_i(\zeta)$,$g_i'(\zeta)$, $g_i''(\zeta)$, $\dots$ of degree at most $n-2$.
Taking the limit of \eqref{eq:matching} as $s \to \infty$ shows $\underset{s \to \infty}{\lim} P_i(s,a_{ij})=0$. A polynomial that satisfies $P_i(a_{ij})=0$ is either a multiple of the polynomial $A_i(\zeta)=\underset{j \neq i}{\prod} (\zeta-a_{ij})$ or it is identically zero. Since $A_i(\zeta)$ has degree exceeding $n-2$, we conclude that $\underset{s \to \infty}{\lim} P_i(s,\zeta)=0$.
It follows that the perturbation expansion must begin
	\begin{align}
	P_i(s,\zeta)=e^{-s\Delta_i}g_i(\zeta)+e^{-s\Delta_i'}g'_i(\zeta)+e^{-s\Delta_i''}g_i''(\zeta)+\dots
	\end{align}
	with $0 < \Delta_i < \Delta_i'<\Delta_i''<\dots$ and $s$-independent polynomials $g_i(\zeta)$,$g_i'(\zeta)$, $g_i''(\zeta)$, $\dots$ of degree at most $n-2$.
	
	The expansion of the matching conditions $P_i(s,a_{ij})=e^{-sr_{ij}}P_j(s,a_{ij})$ at zero-order is the system
	\begin{align}
	g_{i}(a_{ij})=e^{-s(r_{ij}+\Delta_j-\Delta_i)}g_j(a_{ij}), \ \ \ \text{ for } i \neq j.
	\end{align} 
	If $g_{i}(a_{ij}) \neq 0$, then the $s$-independence of the perturbation polynomials implies $r_{ij}+\Delta_j - \Delta_i=0$. However, we make the observation that for a fixed pair of $i,j$ we cannot have both
	\begin{align}
	\Delta_j + r_{ij}-\Delta_i =0 \text{ and } \Delta_i +r_{ij}-\Delta_j =0
	\end{align}
	as that implies $r_{ij}=0$. Therefore, if $g_{i}(a_{ij}) \neq 0$, we must have $g_i(a_{ji})=0$ to avoid this case. This gives $n-1$ many points that $g_i(\zeta)$ must be zero at. However $g_i(\zeta)$ is degree at most $n-2$ so the solution to zero-order is $g_i(\zeta) \equiv 0$. 

By induction, subsequent orders are identical to the zero-order case. We conclude from the perturbation expansion of our section that 
\begin{align*}(P_1(s,\zeta),\dots,P_n(s,\zeta)) \equiv (0,\dots,0), \end{align*}
so that there is no nonzero global section of $L^s(n-2)$.
\end{proof}

\subsection{Inner Product Structure}
The reality conditions \eqref{eq: reality Nahm L M} imply that the frame $U$ of $H^0(\mathcal{C},L^s(n-1))$ for $s \in (0,\infty)$ must satisfy an additional condition of orthonormality with respect to an inner product on $H^0(\mathcal{C},L^s(n-1))$. We now review this inner product structure on $H^0(\mathcal{C},L^s(n-1))$. 

The antipodal map on $\mathbb{P}^1$ lifts to an antiholomorphic involution on $T\mathbb{P}^1$ given by 
\begin{align} \label{eq:real structure on TP1}
\mathcal{J}:(\zeta,\eta) \mapsto (-{1}/{\bar{\zeta}},-{\bar{\eta}}/{\bar{\zeta}^2}).
\end{align}
We have seen that the spectral curve $\mathcal{C} \subset T\mathbb{P}^1$ is invariant under \eqref{eq:real structure on TP1} so the curve has a real structure.
The involution induces an antiholomorphic involution $\sigma$ on the set of line bundles $\text{Pic}(\mathcal{C})$ where $\sigma$ acts on a section $u$ by
\begin{align}
\sigma(u(\zeta,\eta))= \overline{u(-1/\bar{\zeta},-\bar{\eta}/\bar{\zeta}^2)}.
\end{align}

The inner product $\langle \ , \ \rangle$ on $H^0(\mathcal{C},L^s(n-1))$ introduced by Hitchin  \cite[pp.179-181]{hitchin_construction_1983} arises in the following manner. Given sections $u,v \in H^0(\mathcal{C},L^s(n-1))$, then $u\sigma(v)$ is a section of $\mathcal{O}_\mathcal{C}(2n-2)$ and can be written uniquely in the North patch as
\begin{align} \label{eq:origin of inner product}
	u\sigma(v)^N=c_0(s)\eta^{n-1}+c_1(s,\zeta)\eta^{n-2}+\dots+c_{n-1}(s,\zeta),
\end{align}
for $c_i(s,\zeta)$ a degree $2i$ polynomial in $\zeta$.
The inner product is then defined by setting
\begin{align}
	\langle u, v \rangle:=c_0(s)
\end{align}
for $c_0(s)$ of \eqref{eq:origin of inner product}, which is independent of $\zeta$.

The $c_0$ term defining the inner product $\langle u,v \rangle$ on $H^0(\mathcal{C},L^s(n-1))$ may be written in terms of the corresponding row of polynomials \eqref{eq:QandURelationship} as in \cite[Eqn 4.8]{bielawski_reducible_2007}. Given $p=(P_1(\zeta),\dots,P_n(\zeta))$ and $r=(R_1(\zeta),\dots,R_n(\zeta))$ of polynomials satisfying \eqref{eq:PolynomialRowsBundle}, Lagrangian interpolation for \eqref{eq:origin of inner product} using the fiber value of $u\sigma(v)$ above $\zeta \in \mathbb{P}^1$ on each sheet $\mathcal{C}_i$ gives that
\begin{align} \label{eq:inner product polynomials}
\langle p, r \rangle=(-\zeta)^{n-1} \overset{n}{\underset{i=1}{\sum}} \frac{P_i(\zeta) \overline{R_i(-1/\bar{\zeta})}}{\underset{i \neq j}{\prod} (p_i(\zeta) - p_j(\zeta))}.
\end{align}

\subsection{The line bundles}
We conclude this section with some comments on the line bundles in our construction. For a nonsingular curve the
identification of the Picard group and the Jacobian leads to describing the flow of the integrable system in terms of a flow on the Jacobian. For the $su(2)$ charge $n$-monopole the
curve is smooth of genus $(n-1)^2$ and the Jacobian is (a complex torus) of this dimension. When a curve degenerates
the Jacobian is replaced by a generalized Jacobian, generically a product of $\mathbb{C}$'s,  $\mathbb{C}\sp\ast$'s and a Jacobian of smaller genus. For the case at hand the normalization of our curve is the sum of $n$ $\mathbb{P}\sp1$'s and line bundles are restricted by
the matching conditions (\ref{eq:PolynomialRowsBundle}). To get a line bundle on
our curve $\mathcal{C}$ we choose line bundles $\mathcal{L}_i$  on each $\mathbb{P}\sp1_i$ 
and then identify the fibre of  $\mathcal{L}_i$  at $p_i$ with the fibre of  $\mathcal{L}_j$  at $p_j$
This identification is only determined up to scaling by an element of $\mathbb{C}$. Now consider first the
case $n=2$. Here if we choose two different identifications, differing say by a scaling factor of $\lambda$, we get the same line bundle on $\mathcal{C}$ because we can apply the automorphism $1\times\lambda$
 to the line bundle $\mathcal{L}_1\coprod \mathcal{L}_2$  on 
 $\mathbb{P}\sp1_1\coprod\mathbb{P}\sp1_2$  and this will take the first identification to the second.
 For the general case we have the possibility of $n-1$ such scalings. Here there are $n(n-1)$ points
 of intersection and so we obtain $ n(n-1) -(n-1) =(n-1)^2$ worth of scalings. 
These scalings account for the different line bundles (our generalized Jacobian).

\section{From Frames to Nahm Solutions}

We have described the process of obtaining an orthonormal frame $U$ for $H^0(\mathcal{C},L^s(n-1))$ with $s \in (0,\infty)$ from a Nahm solution $(T_0,T_1,T_2,T_3)$. In this section, we describe the converse process of how we may associate to $U$
an appropriate Nahm solution $(T_0,T_1,T_2,T_3)$. This is achieved using the following proposition of Bielawski:
\begin{proposition}\cite[Prop 5.1]{bielawski_reducible_2007} \label{prop:BasisSectionsNahmCorrespondence} The set of orthonormal frames $U$ of $H^0(\mathcal{C},L^s(n-1))$ for $s \in (0,\infty)$ is in 1-1 correspondence with the set of solutions $(T_0,T_1,T_2,T_3)$ to Nahm's equations satisfying the boundary conditions 
\begin{align*}
\underset{s \to 0}{\lim} (sT_0,sT_1,sT_2,sT_3)=(0,\frac{i\rho_1}{2},\frac{i\rho_2}{2},\frac{i\rho_3}{2}), \underset{s \to \infty}{\lim} (T_0,T_1,T_2,T_3) \in \text{ad}_{U(n)}(0,i\tau_1,i\tau_2,i\tau_3),
\end{align*} with $\vec{\rho}$ some $n$ dimensional irreducible representation of $\mathfrak{su}(2)$. 
\end{proposition}

Now, suppose we are given an orthonormal frame $U$ of $H^0(\mathcal{C},L^s(n-1))$. Define the Lax pair $(L,M)$ in the North patch by
\begin{align} \label{eq:Nahm Lax pair from U}
\begin{split}
L^N(s,\zeta)&:=U^N(s,\zeta)\left( \begin{smallmatrix} p_1(\zeta) & & & \\ & p_2(\zeta) & & \\ & & \ddots & \\ & & & p_n(\zeta) \end{smallmatrix} \right) U^N(s,\zeta)^{-1}, \\  M^N(s,\zeta)&:=-\frac{dU^N(s,\zeta)}{ds}U^N(s,\zeta)^{-1},
\end{split}
\end{align}
and in the South patch by
\begin{align*}
\begin{split}
L^S(s,\frac{1}{\zeta})&:=U^S(s,\frac{1}{\zeta})\left( \begin{smallmatrix} \frac{p_1(\zeta)}{\zeta^2} & & & \\ & \frac{p_2(\zeta)}{\zeta^2} & & \\ & & \ddots & \\ & & & \frac{p_n(\zeta)}{\zeta^2} \end{smallmatrix} \right)U^S(s,\frac{1}{\zeta})^{-1}, \\ M^S(s,\frac{1}{\zeta})&:=-\frac{dU^S(s,\frac{1}{\zeta})}{ds}U^S(s,\frac{1}{\zeta})^{-1},
\end{split}
\end{align*}
with each $p_{j}(\zeta)=x_j^1+ix_j^2 - 2x_j^3 \zeta - (x_j^1-ix_j^2)\zeta^2$ determined by the configuration of points $x_j=(x_j^1,x_j^2,x_j^3) \in \mathbb{R}^3$ arising from the boundary conditions as $s \to \infty$.
From the transition function of $U$ given in Proposition \ref{defn:EigenlineBundle}, it follows that the Lax pair $(L,M)$ satisfies the desired transition functions \eqref{eq:transition function Nahm Lax}. It is also a simple exercise to see that orthonormality of $U$ implies $(L,M)$ satisfies the reality conditions of \eqref{eq: reality Nahm L M}.

 We shall now give an elementary proof that the Lax pair $(L,M)$ is holomorphic in $\zeta$; this is not immediately obvious because the columns of $U$ fail to be linearly independent at the double points $\zeta=a_{ij}$ of the spectral curve, which means the inverse of $U$ is meromorphic in $\zeta$.
 Moreover we show that $L$ is quadratic in $\zeta$ and $M$ is linear in $\zeta$ from which it follows that $(L,M)$ may be written in the form \eqref{eq:NahmLaxN} with $(T_0,T_1,T_2,T_3)$ satisfying Nahm's equations.

Instead of the frame $U$, it is more convenient to work with the corresponding matrix of polynomials $Q(s,\zeta)$ defined by \eqref{eq:QandURelationship}. $Q(s,\zeta)$ is orthonormal with respect to the inner product \eqref{eq:inner product polynomials} so the inverse $Q^N(s,\zeta)^{-1}$ is
\begin{align}\label{QNinv}
\footnotesize Q^N(s,\zeta)^{-1}=(-\zeta)^{n-1}\left( \begin{smallmatrix} {\overset{n}{\underset{j=2}{\prod}} (p_1(\zeta)-p_j(\zeta))^{-1}} &  & \\  &  \ddots & \\ & & & {\overset{n-1}{\underset{j=1}{\prod}} (p_n(\zeta)-p_j(\zeta))^{-1}} \end{smallmatrix} \right)Q^N(s,-1/\bar{\zeta})^\dagger,
\end{align}
where $\dagger$ denotes the complex conjugate transpose.


\begin{lemma} \label{lemma: L quadratic} $L^N(s,\zeta)$ defined in Equation \eqref{eq:Nahm Lax pair from U} is a quadratic polynomial in $\zeta$. 
\end{lemma}
\begin{proof}	
From (\ref{eq:Nahm Lax pair from U}), (\ref{QNinv}) the $uv$-th entry of $L^N$ is given by 
\begin{align}
L^N_{uv}=\sum_{s=1}^n p_s(\zeta)\frac{ (-\zeta)^{n-1}Q_{us}(\zeta)\overline{Q_{vs}(-1/\bar{\zeta})}}{\overset{n}{\underset{l \neq s}{\prod}}(p_s(\zeta) - p_l(\zeta))}.
\end{align}
The possible poles of $L_{uv}^N$ at $\zeta=a_{ij}$ arise from the two terms
 \begin{align} \label{eq:possible pole L} p_i(\zeta) \frac{(-\zeta)^{n-1}Q_{ui}(\zeta)\overline{Q_{vi}(-1/\bar{\zeta})}}{\underset{ l \neq i}{\prod}(p_i(\zeta) - p_l(\zeta))}+p_j(\zeta) \frac{(-\zeta)^{n-1}Q_{uj}(\zeta)\overline{Q_{vj}(-1/\bar{\zeta})}}{\underset{l \neq j}{\prod}(p_j(\zeta) - p_l(\zeta))}.
\end{align}
We have $p_i(a_{ij})=p_j(a_{ij})$, and as  $-1/\overline{a_{ij}}=a_{ji}$, the matching conditions \eqref{eq:PolynomialRowsBundle} imply $Q_{ui}(a_{ij})\overline{Q_{vi}(a_{ji})}=Q_{uj}(a_{ij})\overline{Q_{vj}(a_{ji})}$ so that we may factor \eqref{eq:possible pole L} at $\zeta=a_{ij}$ as
\begin{align}
p_i(a_{ij})(-a_{ij})^{n-1}Q_{ui}(a_{ij})\overline{Q_{vi}(a_{ji})} \bigg(  \frac{1}{\underset{ l \neq i}{\prod}(p_i(\zeta) - p_l(\zeta))}+ \frac{1}{\underset{l \neq j}{\prod}(p_j(\zeta) - p_l(\zeta))} \bigg) \bigg|_{\zeta=a_{ij}}.
\end{align}

The latter factor appears as the only possible poles of $\zeta=a_{ij}$.
Lagrange interpolation yields
\begin{align}\label{laginterp}
\sum_{i=1}^n \frac{1}{\underset{j \neq i}{\prod}(p_i(\zeta)-p_j(\zeta))}=0.
\end{align}
Since this sum is equal to zero, there cannot be a pole at $\zeta=a_{ij}$; consequently $L^N_{uv}$ is holomorphic in $\zeta$. Similarly, $L^S_{uv}$ is holomorphic in $1/\zeta$. The transition relation \eqref{eq:transition function Nahm Lax} of $L$ implies $L^N_{uv}$ is a quadratic polynomial in $\zeta$.
\end{proof}

\begin{lemma} \label{lemma: M linear} $M^N(s,\zeta)$ is a linear function in $\zeta$.
\end{lemma}
\begin{proof}
The $uv$-th entry of $M^N$ is similarly given by
\begin{align}
M^N_{uv}=\sum_{s=1}^{n} \frac{\bigg(-\dot{Q}_{us}(\zeta)+Q_{us}(\zeta)h^+_{s}(\zeta)\bigg)(-\zeta)^{n-1}\overline{Q_{vs}(-1/\overline{\zeta})}}{\overset{n}{\underset{l\neq s}{\prod}} (p_s(\zeta) - p_l(\zeta))}.
\end{align}

The possible poles of $M^N_{uv}$ are at $\zeta=a_{ij}$, arising from the two terms
\begin{multline}\label{eq:possible pole M}
\frac{\left(-\dot{Q}_{ui}(\zeta)+Q_{ui}(\zeta)h^+_{i}(\zeta)\right)(-\zeta)^{n-1}\overline{Q_{vi}(-1/\overline{\zeta})}}{\underset{l \neq i}{\prod} (p_i - p_l)} \\ +\frac{\big(-\dot{Q}_{uj}(\zeta)+Q_{uj}(\zeta)h^+_{j}(\zeta)\big)(-\zeta)^{n-1}\overline{Q_{vj}(-1/\overline{\zeta})}}{\underset{l \neq j}{\prod} (p_j - p_l)}.
\end{multline}
We have $r_{ij}=h_j^+(a_{ij})-h_i^+(a_{ij})$ as is easy to check, so the matching conditions \eqref{eq:PolynomialRowsBundle} again make the numerators equal when $\zeta=a_{ij}$, and we may factor \eqref{eq:possible pole M} at $\zeta=a_{ij}$ as
\begin{align}
\begin{split}
\bigg(-\dot{Q}_{ui}(a_{ij})+Q_{ui}(a_{ij})& h_i^+(a_{ij})\bigg)(-a_{ij})^{n-1} \\ &\times \overline{Q_{vi}(a_{ji})}\bigg(  \frac{1}{\underset{ l \neq i}{\prod}(p_i - p_l)}+ \frac{1}{\underset{l \neq j}{\prod}(p_j - p_l)} \bigg) \bigg|_{\zeta=a_{ij}}.
\end{split}
\end{align}
The latter factor appears as the only possible poles of $\zeta=a_{ij}$ and again the Lagrange interpolation (\ref{laginterp})
shows this factor cannot contain any poles at $\zeta=a_{ij}$.

We conclude $M^N_{uv}$ is holomorphic in $\zeta$, and similarly $M^S_{uv}$ is holomorphic in $1/\zeta$. The transition function \eqref{eq:transition function Nahm Lax} of $M$ implies $M^N_{uv}$ is a linear function of $\zeta$.
\end{proof}

We address the boundary behavior of $(L,M)$ as $s$ approaches $0$. At $s=0$, the matching conditions \eqref{eq:PolynomialRowsBundle} at the double points $a_{ij}$ of the spectral curve become
\begin{align}
Q_{i}(a_{ij})=Q_j(a_{ij}).
\end{align}
This is a section $q=(Q_1(\zeta),\ldots,Q_n(\zeta))$ of the pullback of $\mathcal{O}_{\mathbb{P}^1}(n-1)$ to the spectral curve $\mathcal{C}$, which we denote as $\mathcal{O}_\mathcal{C}(n-1)$; each $Q_i(\zeta)$ has degree
$\leq n-1$.
Hitchin \cite[4.5,5.1]{hitchin_construction_1983} showed quite generally that
each section of $\mathcal{O}_\mathcal{C}(n-1)$ is $r(\zeta)(1, 1,\ldots ,1)$ for some polynomial $r(\zeta)$ of degree $\leq n-1$; in the next section we will give a simple proof of this for our situation. Since these rows are $\mathbb{C}[\zeta]$ multiples of the same row, any $n \times n$ matrix with these rows is not invertible. 
Thus, at $s=0$, the $n \times n$ matrix $Q^N(s,\zeta)$ fails to be invertible and so $L$ and $M$ have a pole at $s=0$. The Lax equation of $(L,M)$ shows this pole must be a simple pole and the residue must give a representation of $\mathfrak{su}(2)$. By \cite[Eq.(5.17)]{hitchin_construction_1983}, this representation is the rank $n$ maximal representation of $\mathfrak{su}(2)$.

\section{Construction of Polynomials} \label{section:construction of polynomials}
To construct a solution $(T_0,T_1,T_2,T_3)$ to Nahm's equations, we seek the unique orthogonal basis for the space of $n$-tuplets \begin{align*}q_l=\left(Q_1\sp{(l)}(\zeta),Q_2\sp{(l)}(\zeta),\dots,Q_n\sp{(l)}(\zeta)\right) \end{align*} of degree $\leq n-1$ polynomials satisfying the following \hypertarget{PolynomialConditions}{conditions}:
\begin{enumerate} \label{list:PolynomialConditions}
\item[A1.] At each of the double points $a_{ij}$ of the spectral curve the polynomials satisfy the matching conditions $Q_i\sp{(l)}(a_{ij}) = e^{-sr_{ij}}Q_j\sp{(l)}(a_{ij})$. 
\item[A2.] The $l$-th basis element $q_l=(Q_1\sp{(l)}(\zeta),Q_2\sp{(l)}(\zeta),\dots,Q_n\sp{(l)}(\zeta))$ has
\begin{itemize}
\item
$Q_i\sp{(l)}(\infty)=0$, i.e. $\deg Q_i\sp{(l)}(\zeta) <n-1$ for $i <l$,

\item $ Q_i\sp{(l)}(0)=0$, for $i >l$,
\item $Q_l\sp{(l)}(\zeta)$ monic of degree $n-1$.
\end{itemize}

\end{enumerate}

This choice of orthogonal basis was proposed by Bielawski \cite[Proposition 2.2]{bielawski_reducible_2007}. Here, $q_l$ vanishes at the points above the fixed point $0 \in \mathbb{P}^1$ for all sheets before the $\mathcal{C}_l$ sheet, then vanishes at the points above the antipodal point $\infty \in \mathbb{P}^1$ for all sheets after $\mathcal{C}_l$. One may see from the inner product \eqref{eq:inner product polynomials}, which is independent of $\zeta$, that the basis is indeed orthogonal.

The Nahm solution constructed from \hyperlink{PolynomialConditions}{(A)} satisfies the boundary conditions
\begin{align} \label{eq:NahmBddConstructed}
\lim_{s \to 0} sT_j(s)=\frac{i\rho_j}{2}, 
 \ \ \ \ \ \ \lim_{s \to \infty} (T_1(s),T_2(s),T_3(s))= (i\tau_1,i\tau_2,i\tau_3),
\end{align} 
where $\rho$ is some $n$-dimensional irreducible representation of $SU(2)$. This solution differs from $(T_0,T_1,T_2,T_3)$ satisfying \eqref{eq:NahmBddConditionNoT0} by only a constant gauge transformation $g_0$.

In this section, we introduce two linear systems for constructing \hyperlink{PolynomialConditions}{(A)}. We then introduce a method to obtain its perturbative expansion for large $s$.

\subsection{A Direct Approach}

There is a direct approach to the linear system of polynomials satisfying the matching conditions.
For a row of polynomials $(Q_1,\dots,Q_n)$ (we drop the row index $l$ for ease of notation) we can simply take as unknowns all the coefficients 
$\{q_{ij}: \ i=0,\dots,n-1, j=1,\dots,n\} $ of  \begin{align*}Q_1(\zeta)&=q_{01}+q_{11}\zeta+\dots+q_{(n-1)1}\zeta^{n-1}, \\ \vdots \\ Q_n(\zeta)&=q_{0n}+q_{1n}\zeta+\dots+q_{(n-1)n}\zeta^{n-1}, \end{align*}
and subject $\{q_{ij}\}$ to the conditions \hyperlink{PolynomialConditions}{(A)}.
For each such row we may form the $n^2\times 1$ column vector 
$\mathfrak{q}:=(q_{01},q_{02},\dots,q_{0n},\dots,q_{(n-1)1},q_{(n-1)n},\dots, q_{(n-1)n})^T $
with  all of the unknown coefficients. Then
$[(1,\zeta,\ldots,\zeta\sp{n-1})\otimes \hat e_i]\mathfrak{q}=Q_i(\zeta) $
where $\hat e_i$ is the $1\times n$ row vector with a $1$ in the $i$-th place and
zero otherwise. The matching condition
$Q_i(a_{ij})=e^{-sr_{ij}}Q_j(a_{ij})$ then
takes the form 
$$[ (1,a_{ij},a_{ij}^2,\ldots,a_{ij}^{n-1})\otimes (e^{sr_{ij}/2}\hat{e}_i - e^{-sr_{ij}/2}\hat{e}_j)]\,\mathfrak{q}=0.$$
For each of the $n$ basis elements $q_l$ we may construct the
corresponding column vectors $\mathfrak{q}_l$ and assemble these columns to form a ${n^2 \times n}$ matrix
$\mathcal{P}$ and the matching conditions 
then yield the linear system
\begin{align} \label{eq:pedestrian linear system}
\Xi\, \mathcal{P} =0,\end{align}
where $\Xi$ is the $n(n-1)\times n^2$ matrix whose rows, labelled by $(ij)$,
are $$(1,a_{ij},a_{ij}^2,\ldots,a_{ij}^{n-1})\otimes (e^{sr_{ij}/2}\hat{e}_i - e^{-sr_{ij}/2}\hat{e}_j).$$

\begin{proposition} \label{prop:pedestrian orthogonal basis} 
The orthogonal basis $\{q_1,\dots,q_n \}$ for $s \in (0,\infty)$, with $$q_l=\left(Q_1\sp{(l)}(\zeta),Q_2\sp{(l)}(\zeta),\ldots,Q_n\sp{(l)}(\zeta)\right)$$ for $Q_i\sp{(l)}(\zeta)$ a monic polynomial of maximal degree $n-1$ and satisfying the vanishing conditions $Q_i\sp{(l)}(\infty)=0$ for $i <l$ and $Q_i\sp{(l)}(0)=0$ for $i >l$, is constructible from the linear system \eqref{eq:pedestrian linear system}.
\end{proposition}
\begin{proof}
	Let $\mathcal{P}$ be the ${n^2 \times n}$ be the matrix of unknowns where all of the coefficients in $q_l$ are placed in column $\mathfrak{q}_l$ described above. The vanishing condition $Q\sp{(l)}_i(\infty)=0$, i.e. $Q\sp{(l)}_i(\zeta)$ has degree (at most) $n-2$, for $i<l$ and so the bottom $n$ rows of $\mathcal{P}$ is a lower triangular matrix $\hat{L}$; the monic condition means that $\hat L$ has $1$'s along the diagonal. The vanishing condition $Q\sp{(l)}_i(0)=0$ for $i>l$ makes the top $n$ rows of $\mathcal{P}$ an upper triangular matrix $U$. The linear system $\Xi\, \mathcal{P}=0$ of \eqref{eq:pedestrian linear system} thus takes the form
	\begin{align}
	\begin{pmatrix} a& B& c \end{pmatrix} \begin{pmatrix} U \\ \tilde{P} \\ \hat{L} \end{pmatrix}=0.
	\end{align}
	which, provided the matrix $\begin{pmatrix}a& B \end{pmatrix}$ is invertible, may be solved by writing,
	\begin{align*}
	\begin{pmatrix} a & B \end{pmatrix} \begin{pmatrix} U \\ \tilde{P} \end{pmatrix} &= -c \hat{L},  \quad
	\begin{pmatrix} a & B \end{pmatrix} \begin{pmatrix} U \\ \tilde{P} \end{pmatrix} \hat{L}^{-1}&=-c, \quad
	\begin{pmatrix} U\hat{L}^{-1} \\ \tilde{P}\hat{L}^{-1} \end{pmatrix} &=-\begin{pmatrix} a & B \end{pmatrix}^{-1} c.
	\end{align*}
	By taking the $UL$ decomposition \cite{trefethen1997numerical} of the top $n$ rows of $-\begin{pmatrix} a & B \end{pmatrix}^{-1} c$ one solves for the unknowns $U$ and $\hat{L}^{-1}$;  then applying $\hat{L}$ to the remaining rows of $-\begin{pmatrix} a & B \end{pmatrix}^{-1} c$ one solves for $\tilde{P}$.
	
	Finally the invertibility of the 
	$n(n-1)\times n(n-1)$ matrix
$$\begin{pmatrix}a& B \end{pmatrix}
=\left( (1,a_{ij},a_{ij}^2,\ldots,a_{ij}^{n-2})\otimes (e^{sr_{ij}/2}\hat{e}_i - e^{-sr_{ij}/2}\hat{e}_j) \right)$$ 
is equivalent to $h^0(S,L^s(n-2))=0$, for $s\in(0,\infty)$, which was proved in Proposition \ref{prop: dimension of sections}: this block of $\Xi$ encodes the matching conditions for the polynomials corresponding to sections of $L^s(n-2)$. We shall give an alternate description of the basis described here in Proposition \ref{prop:LP orthonormal basis}.
\end{proof}

At this stage we see Hitchin's result regarding the sections at $s=0$. The null space of the matrix $(a,B)\big|_{s=0}$ is spanned by the $(n-1)$ vectors of the form $(0,\ldots,0,1,0\ldots,0)\sp{T}\otimes(1,1,\ldots,1)\sp{T}$ and these yield  the sections $r(\zeta)(1,1,\ldots,1)$
of $\mathcal{O}_\mathcal{C}(n-1)$ of Hitchin.

\subsection{A Lagrange-Interpolation Approach} \label{subsection: Lagrangian Approach}
We may also use Lagrange interpolation to obtain an equivalent linear system for the construction of the collection of polynomials satisfying \hyperlink{PolynomialConditions}{(A)}. This method was used by Lamy-Poirier in \cite{lamy-poirier_dirac_2015} to impose algebraic conditions on the unknowns, and in this subsection, we further simplify the algebraic conditions to a linear system and prove its invertibility.  
The advantage of this approach is that the matrix of the linear system is smaller, having size $\frac12 {n(n-1)} \times \frac12 {n(n-1)}$, and the solution to the system is obtained through a single matrix inversion, rather than by matrix inversion followed by an UL decomposition.


The linear system is obtained by the following steps.
\begin{description}
    \item[Step 1]  Write the Lagrange interpolation polynomial for $Q_k(\zeta)$, $k=1,2,\dots,n$ in terms of the variables $Q_j(a_{kj})$ from the matching conditions.
    \item[Step 2] Rewrite the Lagrangian polynomials in terms of only the variables $Q_j(a_{kj})$ with $k<j$ for $j=2,3,\dots,n$.
    \item[Step 3] Evaluate the interpolation polynomials $Q_k(\zeta)$ at $a_{jk}$ for $j<k$ to obtain  $\frac{n(n-1)}{2}$ linear equations for the unknowns $Q_u(a_{vu})$ with $v<u$ for $u=2,3,\dots,n$.
\end{description}
We will now describe these steps in more detail, ending the subsection with the
$n=3$ example.

\medskip
\noindent{\bf Step 1:}
As in \cite{lamy-poirier_dirac_2015}, supposing the values $Q_j(a_{kj})$ for $j \neq k$ are fixed, then Lagrange interpolation for the degree $n-1$ polynomial $Q_k(\zeta)$ using the matching conditions $Q_k(a_{kj})=e^{-sr_{kj}}Q_j(a_{kj})$ for $j \neq k$ gives the degree $n-1$ polynomial
\begin{align} \label{eq:Lagrangian interpolation polynomial}
Q_k(\zeta)=C_k A_k(\zeta)+\sum_{j \neq k} e^{-sr_{kj}} Q_j(a_{kj}) \underset{l \neq k,j}{\prod}\frac{\zeta - a_{kl}}{a_{kj}-a_{kl}}, \ k=1,\dots,n.
\end{align}
Here we define
\begin{equation}\label{defA}
A_k(\zeta):=\overset{n}{\underset{j \neq k}{\prod}} (\zeta - a_{kj}).
\end{equation}
The value $C_k$ is some constant independent of $\zeta$, although we do not need to make it independent of $s$.

\medskip
\noindent{\bf Step 2:}
The interpolation polynomials \eqref{eq:Lagrangian interpolation polynomial} satisfy the matching conditions by construction and, altogether, they utilize the $n(n-1)$ many variables $Q_u(a_{vu})$, $u=1,\dots,n$ with both $v<u$ and $v>u$.  We may however rewrite \eqref{eq:Lagrangian interpolation polynomial} in terms of only $Q_u(a_{vu})$ with $v<u$ for $u=2,\dots,n$.
\begin{lemma}
The variables $Q_u(a_{vu})$ with $v>u$ in the Lagrange interpolation polynomials \eqref{eq:Lagrangian interpolation polynomial} may be expressed in terms of $Q_u(a_{vu})$ with $v <u$.
\end{lemma}
\begin{proof}
The Lagrange interpolation polynomial $Q_1(\zeta)$ is already written in terms of $Q_u(a_{1u})$ for $1<u$, so we obtain $Q_1(a_{j1})$ for $j>1$ in terms of $Q_u(a_{1u})$ by evaluating $Q_1(\zeta)$ at $a_{j1}$.

We use induction for $k=2,\dots,n$. The induction hypothesis is that $Q_p(a_{qp})$ for $p<k$ have been written in terms of $Q_u(a_{vu})$ for $v<u$. Then for $Q_k(\zeta)$ of  \eqref{eq:Lagrangian interpolation polynomial},
\begin{align}
	\begin{split} Q_k(\zeta)=C_kA_k(\zeta) &+ \sum_{j < k} e^{-sr_{kj}} \underset{\text{rewritten by induction hypothesis}}{\underbrace{Q_j(a_{kj})}} \underset{l \neq k,j}{\prod}\frac{(\zeta - a_{kl})}{(a_{kj}-a_{kl})}\\&+\sum_{k <j} e^{-sr_{kj}} \underset{\text{desired variables}}{\underbrace{Q_j(a_{kj})}} \underset{l \neq k,j}{\prod}\frac{(\zeta - a_{kl})}{(a_{kj}-a_{kl})}.
	\end{split} \end{align}
Thus, we may evaluate $Q_k(\zeta)$ at $a_{jk}$ for $j>k$ to obtain $Q_k(a_{jk})$ in terms of $Q_u(a_{vu})$ with $v<u$.
\end{proof}
\medskip
\noindent{\bf Step 3:} The final step is to evaluate the polynomials $Q_k(\zeta)$ at $a_{jk}$ for $j<k$ to give us the linear system in the unknowns $X=\{Q_u(a_{vu})\}$, $v<u$. This may be written as
$$AX=B$$
where $B$ is expressed  linearly in terms of the
$C_k$'s with coefficients which are products of $e\sp{-s r_{ij}}$'s and rational in the
$a_{ij}$'s. This system can be solved.

\begin{lemma}\label{lemma:LP linear system invertible} The $\frac12 {n(n-1)} \times \frac12 {n(n-1)}$ matrix $A$ of the linear system for the unknowns $Q_k(a_{jk})$, $k=1,\dots,n$ with $j<k$ is invertible when $s \in (0,\infty)$.
\end{lemma}
\begin{proof}
		Set all $C_k=0$ so that $Q_k(\zeta)$ are only degree $n-2$ polynomials in $\zeta$ satisfying the matching conditions $Q_i(a_{ij})=e^{-sr_{ij}}Q_j(a_{ij})$. This is the system $AX=0$. A nontrivial solution to $AX=0$ then corresponds to a section belonging to $H^0(S,L^s(n-2))$. But from the proof of Proposition \ref{prop: dimension of sections}, $h^0(S,L^s(n-2))=0$, so that $AX=0$ admits no nontrivial solutions and $A$ is invertible.
\end{proof}

Thus far the coefficients $C_1,C_2,\dots,C_n$  in our Lagrange interpolation polynomials have been free parameters. By rotating the spectral curve if necessary so that the spectral curve does not have double points above $\zeta=0$ and $\zeta=\infty$, the conditions in \hyperlink{PolynomialConditions}{(A2)} fix the parameters $C_1,C_2,\dots,C_n$. For instance, the condition that a polynomial $Q_l(\zeta)$ be monic of maximal degree $n-1$ fixes $C_l=1$;
the conditions $Q_l(\infty)=0$ fixes $C_l=0$;
similarly  $Q_l(0)=0$ fixes $C_l$.
In this way for each of the basis  polynomial $n$-tuples satisfying \hyperlink{PolynomialConditions}{(A)} we may construct the corresponding column $B$
where the column of $X$ consists of the unknowns $Q_k(a_{jk})$ for this $n$-tuple.
Then

\begin{proposition} \label{prop:LP orthonormal basis} The basis of polynomial $n$-tuples satisfying \hyperlink{PolynomialConditions}{(A)} is given by $X=A^{-1}B$, with $A$ and $B$ the linear system obtained from the Lagrangian interpolation polynomials above. 
\end{proposition} 

We will give two examples. The first, for
$n=2$,  illustrates how the conditions on the $C_i$ are fixed. The second, the rank $n=3$ example, is more illustrative of the general case.

\begin{example}
We follow the above steps. For $n=2$ the
Lagrange interpolation polynomials are
rather simple. The first step gives
\begin{align*}
 Q_1(\zeta)&=C_1A_1(\zeta)+e^{-sr_{12}}Q_2(a_{12}), \\
Q_2(\zeta)&=C_2A_2(\zeta)+e^{-sr_{12}}Q_1(a_{21}).
\end{align*}
Suppose we have the first basis element $\mathfrak{q}_1=(Q_1(\zeta),
Q_2(\zeta))$. Then $Q_1(\zeta)$ is monic, and $C_1=1$. We require $Q_2(0)=0$ and
so 
$$C_2= \frac{e^{-sr_{12}} }{a_{21}}\, Q_1(a_{21}).$$
We now have
\begin{align*}
&Q_1(a_{21})=a_{21}-a_{12}+e^{-sr_{12}}Q_2(a_{12}), \\
&Q_2(\zeta)=C_2A_2(\zeta)+e^{-sr_{12}}Q_1(a_{21})
=\frac{e^{-sr_{12}} }{a_{21}}\, Q_1(a_{21})
\left[ \zeta- a_{21}\right] +
e^{-sr_{12}}Q_1(a_{21}).
\intertext{Upon substituting  $Q_1(a_{21})$
we have $Q_2(\zeta)$ defined by $Q_2(a_{12})$. Solving for this yields}
Q_2(a_{12})&= \frac{a_{12}[a_{12}-a_{21}]
e^{-sr_{12}} }{a_{21}-a_{12} e^{-2sr_{12}}},
\qquad
Q_2(\zeta)= \frac{\zeta[a_{12}-a_{21}]
e^{-sr_{12}} }{a_{21}-a_{12} e^{-2sr_{12}}}.
\end{align*}

\end{example}

\begin{example}
 Now $n=3$. We again
follow the above steps. We aim to write the linear system obtained by Lagrange
interpolation to solve for a polynomial $3$-tuple $(Q_1(\zeta),Q_2(\zeta),Q_3(\zeta))$ that satisfies the matching conditions  in the form
$$
AX=
\begin{pmatrix} A_{11} & A_{12}&A_{13} \\ 
		A_{21} &   A_{22} & A_{23} \\
	A_{31} & A_{32} & A_{33}\end{pmatrix}
\begin{pmatrix}
Q_2(a_{12}) \\ Q_3(a_{13}) \\Q_3(a_{23})
\end{pmatrix}=
\begin{pmatrix} b_1 \\ b_2 \\ b_2 \end{pmatrix}=B.
$$

The Lagrange interpolation polynomials of \eqref{eq:Lagrangian interpolation polynomial} are
	\begin{align*}
	\begin{split}
	 Q_1(\zeta)&=C_1A_1(\zeta)+e^{-sr_{12}}Q_2(a_{12})\frac{\zeta-a_{13}}{a_{12}-a_{13}} +e^{-sr_{13}}Q_3(a_{13})\frac{\zeta-a_{12}}{a_{13}-a_{12}}, \\
	 Q_2(\zeta)&=C_2A_2(\zeta)+e^{-sr_{12}}Q_1(a_{21})\frac{\zeta-a_{23}}{a_{21}-a_{23}}+e^{-sr_{23}}Q_3(a_{23}) \frac{\zeta -a_{21}}{a_{23}-a_{21}}, \\
	Q_3(\zeta)&=C_3A_3(\zeta)+e^{-sr_{13}}Q_1(a_{31})\frac{\zeta-a_{32}}{a_{31}-a_{32}}+e^{-sr_{23}}Q_2(a_{32})\frac{\zeta-a_{31}}{a_{32}-a_{31}}.
	\end{split}
	\end{align*}
	
One sees that all six values of $Q_u(a_{vu})$ for $v<u$ and $v>u$ are present in this system of polynomials. The three values $Q_u(a_{vu})$ for $v>u$ may be written in terms of the other three $Q_u(a_{vu})$ for $v<u$ by evaluating the interpolation polynomial for $Q_k(\zeta)$ at $a_{jk}$ for $j>k$ and using recursion.
	\begin{align*}
	\begin{split}
	Q_1(a_{21})&=C_1A_1(a_{21})+e^{-sr_{12}}Q_2(a_{12})\frac{a_{21}-a_{13}}{a_{12}-a_{13}}+e^{-sr_{13}}Q_3(a_{13})\frac{a_{21}-a_{12}}{a_{13}-a_{12}}, \\
	Q_1(a_{31})&=C_1A_1(a_{31})+e^{-sr_{12}}Q_2(a_{12})\frac{a_{31}-a_{13}}{a_{12}-a_{13}}+e^{-sr_{13}}Q_3(a_{13})\frac{a_{31}-a_{12}}{a_{13}-a_{12}}, \\
	Q_2(a_{32})&=C_2A_2(a_{32})+e^{-sr_{12}}Q_1(a_{21})\frac{a_{32}-a_{23}}{a_{21}-a_{23}}+e^{-sr_{23}}Q_3(a_{23})\frac{a_{32}-a_{21}}{a_{23}-a_{21}}.
	\end{split}
	\end{align*}
	
Using these the interpolation polynomials for $Q_k(\zeta)$, $k=2,3$, may be expressed in terms of the three unknowns $Q_2(a_{12}),Q_3(a_{13}),Q_3(a_{23})$: 
\begin{align*}
Q_2(\zeta)&= C_2A_2(\zeta)+e^{-sr_{12}}\frac{\zeta-a_{23}}{a_{21}-a_{23}}\bigg(C_1A_1(a_{21})+e^{-sr_{12}}Q_2(a_{12})\frac{a_{21}-a_{13}}{a_{12}-a_{13}} \\ &\qquad +e^{-sr_{13}}Q_3(a_{13})\frac{a_{21}-a_{12}}{a_{13}-a_{12}}\bigg)  +e^{-sr_{23}}Q_3(a_{23})\frac{\zeta-a_{21}}{a_{23}-a_{21}}.
\\
Q_3(\zeta)&=C_3A_3(\zeta) \\ +e^{-sr_{13}}&\frac{\zeta-a_{32}}{a_{31}-a_{32}}
\left(C_1A_1(a_{31})+e^{-sr_{12}}Q_2(a_{12})\frac{a_{31}-a_{13}}{a_{12}-a_{13}}+e^{-sr_{13}}Q_3(a_{13})\frac{a_{31}-a_{12}}{a_{13}-a_{12}}\right) \\ +e^{-sr_{23}}&\frac{\zeta-a_{31}}{a_{32}-a_{31}}\Bigg(C_2A_2(a_{32}) +e^{- sr_{12}}\frac{a_{32}-a_{23}}{a_{21}-a_{23}}\bigg(C_1A_1(a_{21}) +e^{-sr_{12}}Q_2(a_{12})\frac{a_{21}-a_{13}}{a_{12}-a_{13}} \\
&+e^{-sr_{13}}Q_3(a_{13})\frac{a_{21}-a_{12}}{a_{13}-a_{12}}\bigg) +e^{-sr_{23}}Q_3(a_{23})\frac{a_{32}-a_{21}}{a_{23}-a_{21}}\Bigg).
\end{align*}
Now evaluating $Q_2(\zeta)$ at $\zeta=a_{12}$ and $Q_3(\zeta)$ at $\zeta=a_{13},a_{23}$ we obtain the desired linear system in the three unknowns $Q_2(a_{12})$, $Q_3(a_{13})$, $Q_3(a_{23})$. 
	\begin{footnotesize}
	\begin{align*}
	Q_2(a_{12}) &= C_2A_2(a_{12}) \\ &\qquad+e^{-sr_{12}}\frac{a_{12}-a_{23}}{a_{21}-a_{23}}\left(C_1A_1(a_{21})+e^{-sr_{12}}Q_2(a_{12})\frac{a_{21}-a_{13}}{a_{12}-a_{13}}+e^{-sr_{13}}Q_3(a_{13})\frac{a_{21}-a_{12}}{a_{13}-a_{12}}\right) \\ &\qquad +e^{-sr_{23}}Q_3(a_{23})\frac{a_{12}-a_{21}}{a_{23}-a_{21}}.\\
		Q_3(a_{13})&=C_3A_3(a_{13})\\ &\qquad+e^{-sr_{13}} \frac{a_{13}-a_{32}}{a_{31}-a_{32}}\big(C_1A_1(a_{31})+e^{-sr_{12}}Q_2(a_{12})\frac{a_{31}-a_{13}}{a_{12}-a_{13}}+e^{-sr_{13}}Q_3(a_{13})\frac{a_{31}-a_{12}}{a_{13}-a_{12}}\big) \\ &\qquad +e^{-sr_{23}}\Bigg(C_2A_2(a_{32})+e^{-sr_{12}}\frac{a_{32}-a_{23}}{a_{21}-a_{23}}\bigg[C_1A_1(a_{21})+e^{-sr_{12}}Q_2(a_{12})\frac{a_{21}-a_{13}}{a_{12}-a_{13}} \\ &\qquad\qquad\qquad\qquad\qquad +e^{-sr_{13}}Q_3(a_{13})\frac{a_{21}-a_{12}}{a_{13}-a_{12}}\bigg]  +e^{-sr_{23}}Q_3(a_{23})\frac{a_{32}-a_{21}}{a_{23}-a_{21}}\Bigg) \frac{a_{13}-a_{31}}{a_{32}-a_{31}}.\\
	Q_3(a_{23})&=C_3A_3(a_{23})+e^{-sr_{13}}\frac{a_{23}-a_{32}}{a_{31}-a_{32}}\bigg(C_1A_1(a_{31})+e^{-sr_{12}}Q_2(a_{12})\frac{a_{31}-a_{13}}{a_{12}-a_{13}} \\ &\qquad\qquad\qquad\qquad\qquad\qquad\qquad\qquad
	+e^{-sr_{13}}Q_3(a_{13})\frac{a_{31}-a_{12}}{a_{13}-a_{12}}\bigg)
	\\ &\qquad+e^{-sr_{23}}\frac{a_{23}-a_{31}}{a_{32}-a_{31}}\Bigg(C_2A_2(a_{32})+e^{-sr_{12}}\frac{a_{32}-a_{23}}{a_{21}-a_{23}}\bigg[C_1A_1(a_{21})\\ &\qquad+e^{-sr_{12}}Q_2(a_{12})\frac{a_{21}-a_{13}}{a_{12}-a_{13}} +e^{-sr_{13}}Q_3(a_{13})\frac{a_{21}-a_{12}}{a_{13}-a_{12}}\bigg] +e^{-sr_{23}}Q_3(a_{23})\frac{a_{32}-a_{21}}{a_{23}-a_{21}}\Bigg).
		\end{align*}
	\end{footnotesize}
Rearranging these yield
		\begin{align*}
		A_{11}&=-1+e^{-2sr_{12}}\frac{a_{21}-a_{13}}{a_{12}-a_{13}}\frac{a_{12}-a_{23}}{a_{21}-a_{23}}, \\
		A_{12}&= e^{-s(r_{12}+r_{13})}\frac{a_{21}-a_{12}}{a_{13}-a_{12}}\frac{a_{12}-a_{23}}{a_{21}-a_{23}}, \\
		A_{13}&= e^{-sr_{23}}\frac{a_{12}-a_{21}}{a_{23}-a_{21}} , \\
		A_{21}&= e^{-s(r_{13}+r_{12})}\frac{a_{31}-a_{13}}{a_{12}-a_{13}}\frac{a_{13}-a_{32}}{a_{31}-a_{32}} +e^{-s(r_{23}+2r_{12})}\frac{a_{32}-a_{23}}{a_{21}-a_{23}}\frac{a_{21}-a_{13}}{a_{12}-a_{13}}\frac{a_{13}-a_{31}}{a_{32}-a_{31}}, \\
		A_{22}&= -1+ e^{-2sr_{13}}\frac{a_{31}-a_{12}}{a_{13}-a_{12}}\frac{a_{13}-a_{32}}{a_{31}-a_{32}} +e^{-s(r_{23}+r_{12}+r_{13})}\frac{a_{32}-a_{23}}{a_{21}-a_{23}}\frac{a_{21}-a_{12}}{a_{13}-a_{12}}\frac{a_{13}-a_{31}}{a_{32}-a_{31}}, \\
		A_{23}&= e^{-2sr_{23}}\frac{a_{32}-a_{21}}{a_{23}-a_{21}}\frac{a_{13}-a_{31}}{a_{32}-a_{31}} , \\
		A_{31}&=  e^{-s(r_{13}+r_{12})}\frac{a_{23}-a_{32}}{a_{31}-a_{32}}\frac{a_{31}-a_{13}}{a_{12}-a_{13}}+e^{-s(r_{23}+2r_{12})}\frac{a_{21}-a_{13}}{a_{12}-a_{13}}\frac{a_{23}-a_{31}}{a_{32}-a_{31}}\frac{a_{32}-a_{23}}{a_{21}-a_{23}}, \\
		A_{32}&= e^{-2sr_{13}}\frac{a_{23}-a_{32}}{a_{31}-a_{32}}\frac{a_{31}-a_{12}}{a_{13}-a_{12}}+e^{-s(r_{23}+r_{12}+r_{13})}\frac{a_{21}-a_{12}}{a_{13}-a_{12}}\frac{a_{32}-a_{23}}{a_{21}-a_{23}}\frac{a_{23}-a_{31}}{a_{32}-a_{31}}, \\
		A_{33}&= -1 + e^{-2sr_{23}}\frac{a_{32}-a_{21}}{a_{23}-a_{21}}\frac{a_{23}-a_{31}}{a_{32}-a_{31}}.
		\end{align*}
		and
\begin{align*}
b_1&=C_2A_2(a_{12})+e^{-sr_{12}}C_1A_1(a_{21})\frac{a_{12}-a_{23}}{a_{21}-a_{23}},
\\
b_2&=C_3A_3(a_{13})+e^{-sr_{13}}C_1A_1(a_{31})\frac{a_{13}-a_{32}}{a_{31}-a_{32}}+e^{-sr_{23}}C_2A_2(a_{32})\frac{a_{13}-a_{31}}{a_{32}-a_{31}} \\ &\qquad +e^{-s(r_{23}+r_{12})}C_1A_1(a_{21})\frac{a_{32}-a_{23}}{a_{21}-a_{23}}\frac{a_{13}-a_{31}}{a_{32}-a_{31}},
	   \\
	   b_3&=C_3A_3(a_{23})+e^{-sr_{13}}C_1A_1(a_{31})\frac{a_{23}-a_{32}}{a_{31}-a_{32}}+e^{-sr_{23}}C_2A_2(a_{32})\frac{a_{23}-a_{31}}{a_{32}-a_{31}} \\ &\qquad
	   +e^{-s(r_{23}+r_{12})}C_1A_1(a_{21})\frac{a_{23}-a_{31}}{a_{32}-a_{31}}\frac{a_{32}-a_{23}}{a_{21}-a_{23}},
\end{align*}
where $B$ is linear in the $C_i$'s.

\end{example}

\subsection{Perturbative Expansion of Polynomials}
We would like to understand the behavior of the basis of $n$-tuple polynomials satisfying \hyperlink{PolynomialConditions}{(A)} for large $s$. We do this by constructing its approximate solution, which we call its perturbative expansion, for large $s$ in terms of a series in the small parameters $e^{-sr_{ij}}$. The higher-order terms in the series become successively smaller and we give a method for constructing the approximate solution to arbitrary order. 
From this, we see that the matrix of our basis becomes diagonal as $s \to \infty$ so that the corresponding Nahm solution satisfies the boundary condition \eqref{eq:NahmBddConstructed}.

Consider the first element of the basis
${q}_1=(Q_1\sp{(1)},\ldots,Q_n\sp{(1)})$; the story is analogous for other elements (and we will again drop the superscripts for ease of notation).
Set $Q_j'(\zeta):=\lim_{s\rightarrow\infty} Q_j(\zeta)$. The limit as $s \to \infty$ of the matching conditions $Q_{j}(a_{ji})=e^{-sr_{ij}}Q_{i}(a_{ji})$ gives
$
Q'_{j}(a_{ji})=0.
$
In addition to this, 
the conditions of \hyperlink{PolynomialConditions}{(A2)} on ${q}_1$
states that $Q'_{j \neq 1}(\zeta)$ vanishes at $0$ and at $n-1$ many other points $a_{jk}$. Since $Q'_{j \neq 1}(\zeta)$ is a polynomial of degree at most $n-1$, we must have $Q'_{j \neq 1}(\zeta) \equiv 0$.  As $Q'_1(\zeta)$ vanishes at $a_{1k}$ for $k=2,3,\dots,n$ and is a monic polynomial, then $$Q'_1(\zeta)=\overset{n}{\underset{k =2}{\prod}} (\zeta - a_{1k})=A_1(\zeta)$$ 
where $A_i(\zeta)$ was defined in (\ref{defA}).

We are seeking an expansion for the first basis element, which we have called the perturbation expansion for large $s$, with
\begin{align} \label{eq:perturbation expansion}
Q_1(\zeta)&=A_1(\zeta)+e^{-s \Delta_1}q_1(\zeta)+e^{-s\Delta_1'}q_1'(\zeta)+\dots \\
Q_{j \neq 1}(\zeta)&=e^{-s \Delta_j}\zeta q_j(\zeta)+e^{-s \Delta_j'}\zeta q_j'(\zeta)+\dots
\end{align}
with $0<\Delta_k < \Delta_k'<\Delta_k''<\dots$. The polynomials $q_k(\zeta)$ are degree less than or equal to $n-2$ in $\zeta$ and independent of $s$. 

The method of obtaining the perturbation expansion to arbitrary order is given as follows. The zeroth order is $Q_1=A_1(\zeta)$, $Q_{j \neq 1}(\zeta)=0$. By induction, given the expansion at $n$-th order, the $(n+1)$-st order of $Q_1$ is the Lagrangian interpolation polynomial for the values $e^{-sr_{1i}}e^{-s\Delta_i^{(n)}}\tilde{q}^{(n)}_i(a_{1i})$ at the points $a_{1i}$, for $1 < i \leq n$, where $e^{-s\Delta_i^{(n)}}\tilde{q}^{(n)}_i(\zeta)$ is the $n$-th order of $Q_i(\zeta)$. That is, the $(n+1)$-st order of $Q_1$ is given by
\begin{align} \overset{n}{\underset{i=2}{\sum}} e^{-s(\Delta_i^{(n)}+r_{i1})}\tilde{q}_i^{(n)}(a_{1i}) \prod_{k \neq 1,i} \frac{\zeta - a_{1k}}{ a_{1i}-a_{1k}}. \end{align} 
The $(n+1)$-st order of $Q_{j \neq 1}$ is similar, except each term has an extra factor to account for the vanishing conditions on the basis. Explicitly, the $n+1$-st order is
\begin{align} \overset{n}{\underset{i \neq j}{\sum}} e^{-s(\Delta_i^{(n)}+r_{ij})}\tilde{q}_i^{(n)}(a_{ji}) \frac{\zeta}{a_{ji}} \prod_{k \neq j,i} \frac{\zeta - a_{jk}}{ a_{ji}-a_{jk}}. \end{align} 
The points for the Lagrangian interpolation for the $(n+1)$-st order of $Q_j$  comes from the matching conditions 
\begin{align}
Q_j(a_{ji})=e^{-sr_{ij}}Q_i(a_{ji}).
\end{align}

In general, for the $l$-th basis element, the factor $\frac{\zeta}{a_{ji}}$ is either kept or eliminated according to the vanishing conditions. The limit as $s \to \infty$ of the matrix $Q^N(s,\zeta)$ in this basis is then
\begin{align}
    \underset{s \to \infty}{\lim} Q^N(s,\zeta)=\text{diag}\left(A_1(\zeta),A_2(\zeta),\dots,A_n(\zeta)\right).
\end{align}

We illustrate the method by presenting the $n=2,3$ cases.
\begin{example} \label{example:n=2Exact} The first basis element is ${q}_1=(Q_1^{(1)}(\zeta),Q_2^{(1)}(\zeta))$ with
	\begin{center}
		\begin{tabular}{ l |c |c  |c  | l }
			
			Order & Zeroth & First  & Second  & \dots \\
			\hline
			$Q^{(1)}_1(\zeta)=$ & $\zeta-a_{12}$ & 0 & $e^{-s2r_{12}}(a_{21}-a_{12})\frac{a_{12}}{a_{21}}$ & \dots  \\ 
			\hline
			$Q^{(1)}_2(\zeta)=$ & 0 & $e^{-sr_{12}} (a_{21}-a_{12})\frac{\zeta}{a_{21}}$ & 0 & \dots \\
			
		\end{tabular}
	\end{center}
	It is not difficult to continue the expansion and write a closed form. The denominators reproduce the determinant of the matrix $A$ of the previous section.
	\begin{align*}
	Q^{(1)}_1(\zeta)&=\zeta-a_{12}+e^{-s2r_{12}}(a_{21}-a_{12})\frac{a_{12}}{a_{21}}+e^{-s4r_{12}}(a_{21}-a_{12})(\frac{a_{12}}{a_{21}})^2+\dots \\ &=\zeta-a_{12}+\frac{a_{12}(a_{21}-a_{12})}{a_{21}e^{2sr_{12}}-a_{12}},
	\\
	Q^{(1)}_2(\zeta)&=e^{-sr_{12}}(a_{21}-a_{12})\frac{\zeta}{a_{21}}+e^{-3sr_{12}}(a_{21}-a_{12})\frac{a_{12}\zeta}{a_{21}^2}+e^{-5sr_{12}}(a_{21}-a_{12})\frac{a_{12}^2\zeta}{a_{21}^3}+\dots \\ &=\frac{\zeta(a_{21}-a_{12})}{a_{21}e^{sr_{12}}-e^{-sr_{12}}a_{12}}.
	\end{align*}
	The second basis element is ${q}_2=(Q_1^{(2)}(\zeta),Q_2^{(2)}(\zeta))$ with
		\begin{center}
		\begin{tabular}{ l |c |c  |c  | l }
			
			Order & Zeroth & First  & Second  & \dots \\
			\hline
			$Q^{(2)}_1(\zeta)=$ &0  & $e^{-sr_{12}}(a_{12}-a_{21})$ & 0 & \dots  \\ 
			\hline
			$Q^{(2)}_2(\zeta)=$ & $\zeta-a_{21}$ & 0 & $e^{-s2r_{12}}(a_{12}-a_{21})$ & \dots \\
			
		\end{tabular}
	\end{center}
	Again, we may continue the expansion and write a closed form.
	\begin{align*}
	Q^{(2)}_1(\zeta)&=e^{-sr_{12}}(a_{12}-a_{21})+e^{-s3r_{12}}(a_{12}-a_{21})+\dots \\ &=\frac{a_{12}-a_{21}}{e^{sr_{12}}-e^{-sr_{12}}},
	 \\
	Q^{(2)}_2(\zeta)&=\zeta-a_{21}+e^{-s2r_{12}}(a_{12}-a_{21})+e^{-s4r_{12}}(a_{12}-a_{21})+\dots \\ &=\zeta-a_{21}+\frac{a_{12}-a_{21}}{e^{2sr_{12}}-1}.
	\end{align*}
\end{example}

\begin{example}	\label{example:n=3Perturbative}
	We present the $n=3$ case up to second order for the first basis element to illustrate the method. For the remaining elements, we limit our formulas to first order. 
	The first basis element ${q}_1=(Q_1^{(1)}(\zeta),Q_2^{(1)}(\zeta),Q_3^{(1)}(\zeta))$ is
	\begin{center}
		\begin{tabular}{  p{1cm} |P{1cm}  |P{4cm}  |P{6.5cm}  }
			
			Order & Zeroth & First  & Second   \\
			\hline
			$Q^{(1)}_1=$ & $A_1(\zeta)$ & 0 & $e^{-2sr_{12}}A_1(a_{21})\frac{a_{12}}{a_{21}}\frac{a_{12}-a_{23}}{a_{21}-a_{23}}\frac{\zeta-a_{13}}{a_{12}-a_{13}}+e^{-2sr_{13}}A_1(a_{31})\frac{a_{13}}{a_{31}}\frac{a_{13}-a_{32}}{a_{31}-a_{32}}\frac{\zeta-a_{12}}{a_{13}-a_{12}}$ \\ 
			\hline
			$Q^{(1)}_2=$ & 0 & $e^{-sr_{12}}A_1(a_{21})\frac{\zeta}{a_{21}}\frac{\zeta-a_{23}}{a_{21}-a_{23}}$ & $e^{-s(r_{12}+r_{13})}A_1(a_{31})\frac{a_{23}}{a_{31}}\frac{a_{23}-a_{32}}{a_{31}-a_{32}}\frac{\zeta}{a_{23}}\frac{\zeta-a_{21}}{a_{23}-a_{21}}$   \\
			\hline
			$Q^{(1)}_3=$ & 0 & $e^{-sr_{13}}A_1(a_{31})\frac{\zeta}{a_{31}}\frac{\zeta-a_{32}}{a_{31}-a_{32}}$ & $e^{-s(r_{12}+r_{13})}A_1(a_{21})\frac{a_{32}}{a_{21}}\frac{a_{32}-a_{23}}{a_{21}-a_{23}}\frac{\zeta}{a_{32}}\frac{\zeta-a_{31}}{a_{32}-a_{31}}$\\
			
		\end{tabular}
	\end{center}
	
	The second element ${q}_2=(Q_1^{(2)}(\zeta),Q_2^{(2)}(\zeta),Q_3^{(2)}(\zeta))$  to first order is
	\begin{align*}
	Q^{(2)}_1(\zeta)&=e^{-sr_{12}}A_2(a_{12})\frac{\zeta-a_{13}}{a_{12}-a_{13}}, \quad
	Q^{(2)}_2(\zeta)=A_2(\zeta), \quad
	Q^{(2)}_3(\zeta)=e^{-sr_{23}}A_2(a_{32})\frac{\zeta}{a_{32}}\frac{\zeta-a_{31}}{a_{32}-a_{31}}.
	\end{align*}
	
	The third element ${q}_3=(Q_1(\zeta),Q_2(\zeta),Q_3(\zeta))=(Q_1^{(3)}(\zeta),Q_2^{(3)}(\zeta),Q_3^{(3)}(\zeta))$  to first order is
	\begin{align*}
	Q^{(3)}_1(\zeta)&=e^{-sr_{13}}A_3(a_{13})\frac{\zeta-a_{12}}{a_{13}-a_{12}}, \qquad
	Q^{(3)}_2(\zeta)=e^{-sr_{23}}A_3(a_{23})\frac{\zeta-a_{21}}{a_{23}-a_{21}}, \quad
	Q^{(3)}_3(\zeta)=A_3(\zeta).
	\end{align*}
\end{example}

\section{Solutions to Nahm's Equations}
We illustrate the procedure of section \ref{section:TheSpectralApproach} for obtaining rank $n$ Nahm solutions by giving examples of an exact Nahm solution for $n=2$ and a perturbative Nahm solution for $n=3$. We obtained the orthogonal basis $Q(s,\zeta)$ of \hyperlink{PolynomialConditions}{(A)}, consisting of polynomials of degree $\leq n-1$ in $\zeta$, for $n=2$ in Example \ref{example:n=2Exact} and for $n=3$ in Example \ref{example:n=3Perturbative}. For each example, we now find the corresponding Nahm solution.

The Nahm solutions are constructed from the Lax pair $(L,M)$ in terms of $Q(s,\zeta)$ as in Proposition \ref{prop:BasisSectionsNahmCorrespondence}, and we may evaluate $L^N$ and $M^N$ at $\zeta=0$ and obtain
\begin{align} \label{eq:NahmFromLax}
\begin{aligned}
T_1(s)&=\frac{i}{2}(L^N(s,0)+L^N(s,0)^\dagger),  &T_2(s)&=\frac{1}{2}(L^N(s,0)-L^N(s,0)^\dagger), \\ 
T_3(s)&=\frac{i}{2}(M^N(s,0)+M^N(s,0)^\dagger), 
&T_0(s)&=\frac{1}{2}(M^N(s,0)-M^N(s,0)^\dagger).
\end{aligned}
\end{align}

\subsection{n=2 Exact Solution}
The exact solution to Nahm's equations for $n=2$ is known to be given by $T_i=f_i(s)\sigma_i$, for $f_i(s)$ the hyperbolic functions satisfying the Euler top system \cite[Eqn. 4]{Ward:1985ww}.
For example, taking the boundary conditions in \eqref{eq:NahmBddConditionNoT0} to be \begin{align} \label{eq:n=2BoundaryConditions}
\lim_{s \to \infty} (T_1(s),T_2(s),T_3(s)) \in \text{ad}_{U(n)} \left(0,0,i\left(\begin{smallmatrix}
c/2 & 0 \\ 0 & -c/2 \end{smallmatrix}\right)\right) , \ \ \ \  \lim_{s \to 0} sT_j(s)=\frac{i\sigma_j}{2},
\end{align} the unique solution to Nahm's equations is 
\begin{align} \label{eq:n=2HyperbolicSolution}
T_1&=i \begin{pmatrix} 0 & 1 \\ 1 & 0 \end{pmatrix} \frac{c}{2\sinh(cs)}, \ 
T_2 =i \begin{pmatrix} 0 & -i \\ i & 0 \end{pmatrix} \frac{c}{2\sinh(cs)}, \ 
T_3=i \begin{pmatrix} 1 & 0 \\ 0 & -1 \end{pmatrix} \frac{c}{2\tanh(cs)}. 
\end{align}

As an illustration we shall rederive this  solution from the basis in Example \ref{example:n=2Exact}.
The matrix $Q$ of Example \ref{example:n=2Exact} is 
\begin{align}
Q=\begin{pmatrix}\zeta - a_{12} + \frac{a_{12}(a_{21}-a_{12})}{a_{21}e^{2sr_{12}}-a_{12}}&\frac{\zeta(a_{21}-a_{12})}{a_{21}e^{sr_{12}}-a_{12}e^{-sr_{12}}} \\ \frac{a_{12}-a_{21}}{e^{sr_{12}}-e^{-sr_{12}}} & \zeta-a_{21} +\frac{a_{12}-a_{21}}{e^{2sr_{12}}-1}
\end{pmatrix}.
\end{align}
The rows of $Q$ must be normalized with respect to the norm of \eqref{eq:inner product polynomials}. Let $z_{j}:=x_j^1+ix_j^2$ and $z_{jk}:=z_j-z_k$. The norm of the row $\psi_{j}(s,\zeta)=(Q_{1}\sp{(j)},Q_{2}\sp{(j)})$ for $1 \leq j \leq 2$ is
\begin{align}
|| \psi_1 ||^2&= \frac{-1+e^{-2sr_{12}}}{-\bar{z}_{12}(a_{21}-a_{12}e^{-2sr_{12}})}, \quad
|| \psi_2||^2 =\frac{a_{21}-a_{12}e^{-2sr_{12}}}{-\bar{z}_{12}a_{12}a_{21}(1-e^{-2sr_{12}})}.
\end{align}
The resulting Nahm solution from \eqref{eq:NahmFromLax} is
\begin{align}
\begin{split}
&T_1=i\begin{pmatrix} x_1^1 & -\frac{r_{12}}{2\sinh(sr_{12})}\frac{\overline{z_{12}}}{|z_{12}|} \\ -\frac{r_{12}}{2\sinh(sr_{12})}\frac{z_{12}}{|z_{12}|} & x_2^1 \end{pmatrix}, \quad
T_2=i\begin{pmatrix} x_1^2 & \frac{-ir_{12}}{2\sinh(sr_{12})}\frac{\overline{z_{12}}}{|z_{12}|} \\ \frac{ir_{12}}{2\sinh(sr_{12})}\frac{z_{12}}{|z_{12}|} &x_2^2 \end{pmatrix}, \\
&T_3= \frac{i}{2} \left(\begin{matrix}
2x_1^3-\frac{r_{12}^2}{\sinh(sr_{12})(r_{12}\cosh(sr_{12})-x_{12}^3\sinh(sr_{12}))} &
\frac{-r_{12}|z_{12}|}{r_{12}\cosh(sr_{12})-x_{12}^3\sinh(sr_{12})}
\\ \frac{-r_{12}|z_{12}|}{r_{12}\cosh(sr_{12})-x_{12}^3\sinh(sr_{12})} & 2x_2^3+\frac{r_{12}^2}{\sinh(sr_{12})(r_{12}\cosh(sr_{12})-x_{12}^3\sinh(sr_{12}))}\end{matrix}\right), \\
&T_0=\begin{pmatrix}0 & \frac{r_{12}|z_{12}| }{2(r_{12}\cosh(sr_{12})-x_{12}^3\sinh(sr_{12}))} \\ \frac{-r_{12} |z_{12}|}{2(r_{12}\cosh(sr_{12})-x_{12}^3\sinh(sr_{12}))} & 0 \end{pmatrix}.
\end{split}
\end{align}
When taking the point configuration $(0,0,c/2)$ and $(0,0,-c/2)$ of $\mathbb{R}^3$ from the boundary conditions \eqref{eq:n=2BoundaryConditions} for $s \to \infty$, the above solution is
\begin{align}
T_1=i\begin{pmatrix} 0 & 1 \\ 1 & 0 \end{pmatrix}\frac{-c}{2\sinh(cs)}, \quad
T_2=i\begin{pmatrix} 0 & -i \\ i & 0 \end{pmatrix}\frac{c}{2\sinh(cs)}, \quad
T_3=i\begin{pmatrix} 1 & 0 \\ 0 & -1 \end{pmatrix}\frac{-c}{2\tanh(cs)}.
\end{align}
This does not yet satisfy the boundary conditions \eqref{eq:n=2BoundaryConditions} at $s=0$: the constant gauge transformation taking our solution to the one of \eqref{eq:n=2HyperbolicSolution} is simply $g_0=\left( \begin{smallmatrix} 0 & -i \\ i & 0 \end{smallmatrix} \right)$.

\subsection{Perturbative Expansion of the n=3 Solution}
We now give the perturbative solution for large $s$ to Nahm's equations for $n=3$, up to the first order using the perturbative basis of $n$-tuples of polynomials satisfying \hyperlink{PolynomialConditions}{(A)}  we found in Example \ref{example:n=3Perturbative} up to this order.
The perturbation expansion for large $s$ of a Nahm solution is an expansion with the form
\begin{align}
T_j(s)=T_j^{(0)}+e^{-s\Delta_j}T_j^{(1)}+e^{-s\Delta_j'}T_j^{(2)}+\dots
\end{align}
with $0 < \Delta_j < \Delta_j' < \Delta_j'' <\dots$. The matrices $T_j^{(k)}$ are independent of $s$. 

For the point $\vec{x}_j=(x_j^1,x_j^2,x_j^3) \in \mathbb{R}^3$, we again write 
$z_j:=x_j^1+ix_j^2$, $z_{jk}:=z_j-z_k$,
$r_{jk}:=||\vec{x}_j-\vec{x}_k||$ and
set $x_j:=x_j^3$.
The perturbation expansion for large $s$ of $Q$ to the first order is
\begin{align*}
 Q(s,\zeta)= \begin{pmatrix}Q_{11} & Q_{12} &Q_{13} \\ Q_{21} & Q_{22} & Q_{23} \\ Q_{31} & Q_{32} & Q_{33}
\end{pmatrix},
\end{align*}
with
\begin{align*}
\begin{split}
&Q_{11}= (\zeta-a_{12})(\zeta-a_{13}), \\
&Q_{12}=e^{-sr_{12}}(a_{21}-a_{12})(a_{21}-a_{13})\frac{\zeta}{a_{21}}\frac{\zeta-a_{23}}{a_{21}-a_{23}}, \\
&Q_{13}=e^{-sr_{13}}(a_{31}-a_{12})(a_{31}-a_{13})\frac{\zeta}{a_{31}}\frac{\zeta-a_{32}}{a_{31}-a_{32}},
\end{split}
\end{align*}
\begin{align*}
\begin{split}
&Q_{21}=e^{-sr_{12}}(a_{12}-a_{21})(a_{12}-a_{23})\frac{\zeta-a_{13}}{a_{12}-a_{13}}, \\
&Q_{22}=(\zeta-a_{21})(\zeta-a_{23}),\\
&Q_{23}=e^{-sr_{23}}(a_{32}-a_{21})(a_{32}-a_{23})\frac{\zeta}{a_{32}}\frac{\zeta-a_{31}}{a_{32}-a_{31}},
\end{split}
\end{align*}
\begin{align*}
\begin{split}
&Q_{31}=e^{-sr_{13}}(a_{13}-a_{31})(a_{13}-a_{32})\frac{\zeta-a_{12}}{a_{13}-a_{12}},\\
&Q_{32}=e^{-sr_{23}}(a_{23}-a_{31})(a_{23}-a_{32})\frac{\zeta-a_{21}}{a_{23}-a_{21}},\\
&Q_{33}=(\zeta-a_{31})(\zeta-a_{32}).
\end{split}
\end{align*}

We must normalize $Q$ so that each row has norm 1 with respect to the norm \eqref{eq:inner product polynomials}. The norm of the row $\psi_{j}(s,\zeta)=(Q_{1}^{(j)},Q_{j}^{(j)},Q_{j}^{(j)})$ for $1 \leq j \leq 3$, to the first order, is
\begin{align}
|| \psi_1 ||&= \frac{1}{\sqrt{\bar{z}_{12}a_{21}\bar{z}_{13}a_{31}}}, \
|| \psi_2||  =\frac{1}{\sqrt{-\bar{z}_{12}a_{12}\bar{z}_{23}a_{32}}}, \
|| \psi_3||= \frac{1}{\sqrt{\bar{z}_{13}a_{13}\bar{z}_{23}a_{23}}}.
\end{align}
(Note that $\bar{z}_{ji}a_{ij}=x_i-x_j+r_{ij}\ge0$ so these are indeed real.)
Using these the Lax pair $(L,M)$, evaluated at $\zeta=0$, obtained from the normalized $Q$ has the following form
\begin{align*}
L^N(s,0)=\begin{pmatrix} L_{11} & 0 & 0 \\ L_{21} & L_{22} & 0 \\ L_{31} & L_{32} & L_{33} \end{pmatrix},
\qquad
M(s,0)=\begin{pmatrix}M_{11} & 0 & 0 \\ M_{21} & M_{22} & 0 \\ M_{31} & M_{32} & M_{33} \end{pmatrix}. 
\end{align*}

Here
\begin{align*}
L_{11} &=z_1, \\
L_{21} &=e^{-sr_{12}}\frac{|z_{12}|\sqrt{\bar{z}_{13}a_{31}\bar{z}_{23}a_{32}}(a_{21}-a_{12})}{z_{12}} \times \\  
&\qquad\qquad \left(\frac{z_1a_{13}(a_{12}-a_{23})}{z_{13}(a_{12}-a_{13})}-\frac{z_2a_{23}a_{32}(a_{12}-a_{31})}{z_{23}a_{31}(a_{12}-a_{32})} \right), \\
L_{22}&= z_2, \\
L_{31}&=e^{-sr_{13}}\frac{|z_{13}|\sqrt{-\bar{z}_{12}a_{21}\bar{z}_{23}a_{23}}(a_{13}-a_{31})}{z_{13}} \times\\ &\qquad\qquad
\bigg(\frac{z_1a_{12}(a_{13}-a_{32})}{z_{12}(a_{12}-a_{13})}
-\frac{z_3a_{23}a_{32}(a_{13}-a_{21})}{z_{23}a_{21}(a_{13}-a_{23})}\bigg),
\\
L_{32}&=e^{-sr_{23}}\frac{|z_{23}|\sqrt{\bar{z}_{12}a_{12}\bar{z}_{13}a_{13}} (a_{23}-a_{32}) }{z_{23}}\times\\ &\qquad\qquad
\bigg(\frac{z_2a_{21}(a_{23}-a_{31})}{z_{12}(a_{23}-a_{21})}-\frac{z_3a_{13}a_{31}(a_{12}-a_{23})}{z_{13}a_{12}(a_{13}-a_{23})}\bigg), \\
L_{33} &=z_3.
\end{align*}
and
\begin{align*}
M_{11}&=x_1,\\
M_{21}&=e^{-sr_{12}}\frac{|z_{12}|\sqrt{\bar{z}_{13}a_{31}\bar{z}_{23}a_{32}}(a_{12}-a_{21})}{z_{12}} \times \\  
&\qquad\qquad \left(\frac{x_2a_{32}a_{23}(a_{12}-a_{31})}{z_{23}a_{31}(a_{12}-a_{32})}-\frac{(x_1+r_{12})a_{13}(a_{12}-a_{23})}{z_{13}(a_{12}-a_{13})}\right),
\\
M_{22}&=x_2,\\
M_{31}&=e^{-sr_{13}}\frac{|z_{13}|\sqrt{-\bar{z}_{12}a_{21}\bar{z}_{23}a_{23}}(a_{13}-a_{31})}{z_{13}} \times \\  
&\qquad\qquad \bigg( \frac{x_3a_{23}a_{32}(a_{13}-a_{21})}{z_{23}a_{21}(a_{13}-a_{23})}-\frac{(x_1+r_{13})a_{12}(a_{13}-a_{32})}{z_{12}(a_{12}-a_{13})} \bigg), 
\\
M_{32}& = e^{-sr_{23}} \frac{|z_{23}|\sqrt{\bar{z}_{12}a_{12}\bar{z}_{13}a_{13}}(a_{23}-a_{32})}{z_{23}}\times \\  
&\qquad\qquad  \left(\frac{x_3a_{13}a_{31}(a_{12}-a_{23})}{z_{13}a_{12}(a_{23}-a_{13})}-\frac{(x_2+r_{23})a_{21}(a_{23}-a_{31})}{z_{12}(a_{21}-a_{23})}\right),\\
M_{33}&=x_3.
\end{align*}

From (\ref{eq:NahmFromLax})  we obtain the following Nahm solutions:
\begin{align}
\begin{split}
T_1(s)&=\frac{i}{2}\begin{pmatrix} z_1+\bar{z}_1 & \bar{L}_{21} & \bar{L}_{31} \\ L_{21} & z_2 + \bar{z}_2 & \bar{L}_{32} \\ L_{31} & L_{32} & z_3+\bar{z}_3 
\end{pmatrix}+O\left(e^{-\alpha s}\right),\\
T_2(s)&=\frac{1}{2} \begin{pmatrix} z_1 - \bar{z}_1 & - \bar{L}_{21} & -\bar{L}_{31} \\ L_{21} & z_2-\bar{z}_2 & -\bar{L}_{32} \\ L_{31} & L_{32} & z_3-\bar{z}_3\end{pmatrix}+O\left(e^{-\alpha s}\right), \\
T_3(s)&=\frac{i}{2} \begin{pmatrix} 2x_1 & \bar{M}_{21} & \bar{M}_{31} \\ M_{21} & 2x_2 &\bar{M}_{32} \\ M_{31} & M_{32} &2x_3 \end{pmatrix}+O\left(e^{-\alpha s}\right), \\
T_0(s)&=\frac{1}{2} \begin{pmatrix}0 & -\bar{M}_{21} & -\bar{M}_{31} \\ M_{21} & 0 & -\bar{M}_{32} \\ M_{31} & M_{32} & 0 \end{pmatrix}+O\left(e^{-\alpha s}\right),
\end{split}
\end{align}
where $\alpha$ is the minimum of the values $2r_{12}$, $2r_{13}$, and $2r_{23}$.

\section{Dirac Zero Modes}
In this section we review how to obtain normalizable zero modes of the monopole Dirac operator from a basis of polynomial tuples satisfying \hyperlink{PolynomialConditions}{(A)}. The Dirac zero modes allows one to carry out the Nahm transform between monopoles and solutions to Nahm's equations,

\[
\begin{tikzcd}
\text{Monopole}  \arrow[leftrightarrow]{rr}{\text{Nahm Transform}} && [4em] \text{Nahm Solution}.
\end{tikzcd}
\]

For the bundle and self-dual connection $(E,A)$ over the reduced space $X=\mathbb{R}^4 / \Lambda$, we have Dirac operators $\slashed{\mathfrak{D}}_{x^\ast}$ for a family of connections of the trivial bundle $I$ twisted by $x^\ast \in X^\ast$. With respect to the chiral decomposition $S=S^+ \oplus S^-$, the Dirac operator $\slashed{\mathfrak{D}}_{x^\ast}:\Gamma(S \otimes E \otimes I) \to \Gamma(S \otimes E \otimes I)$ has the form
\begin{align*}
    \slashed{\mathfrak{D}}_{x^\ast}=\begin{pmatrix} 0 & \mathfrak{D}_{x^\ast} \\ \mathfrak{D}_{x^\ast}^\dagger & 0 \end{pmatrix},
\end{align*}
and we are interested in the normalizable zero modes of $\mathfrak{D}_{x^\ast}^\dagger$. The operator $\mathfrak{D}_{x^\ast}$ has no normalizable zero modes but, as we shall see following Nahm \cite{nahm_algebraic_1983}, it is an useful operator to consider.

\subsection{Commuting Operators}
Identifying $\mathbb{R}^4$ with $\mathbb{C}\sp2$ by setting $z=x_1+ix_2$ and $w=x_3+ix_4$. The anti-self-dual equation \eqref{eq:ASDequation} can be reformulated using a pair of operators $(\mathcal{L},\mathcal{M})$ with spectral parameter $\zeta \in \mathbb{P}^1$, where
\begin{align} \label{eq:commuting operators}
    \mathcal{L}=D_{\bar{z}}-D_w \zeta, \ \ \ \mathcal{M}=D_{\bar{w}}+D_z \zeta.
\end{align}
There is nothing special about our choice of complex structure; indeed the spectral parameter $\zeta$ parametrizes all complex structures on $\mathbb{R}^4$.
The anti-self-dual equation \eqref{eq:ASDequation} is then equivalent to $(\mathcal{L},\mathcal{M})$ commuting,
 $[\mathcal{L},\mathcal{M}]=0$, for all choices of $\zeta$.

The Dirac operator $\mathfrak{D}$, whose explicit formulation shall be given in the monopole and Nahm contexts later, is related to $(\mathcal{L},\mathcal{M})$ via
\begin{align} \label{eq:Dirac and ASD Operators Relation}
    \begin{pmatrix} 1 & 0 \\ -\zeta & 1 \end{pmatrix} \mathfrak{D} \begin{pmatrix} 1 \\ \zeta \end{pmatrix} = \begin{pmatrix} \mathcal{M} \\ \mathcal{L} \end{pmatrix}.
\end{align}
Thus, if $\chi$ is a solution to the associated linear problem
$\mathcal{L} \chi =0$ and $ \mathcal{M} \chi =0$,
then, since the matrix $\begin{pmatrix} 1 & 0 \\ -\zeta & 1 \end{pmatrix}$ is invertible, we have $\mathfrak{D} \begin{pmatrix} 1 \\ \zeta  \end{pmatrix} \otimes \chi=0.$

\subsection{Nahm-sided Zero Modes}
We use the orthonormal basis $U$ of $H^0(S,L^s(n-1))$ and take advantage of 
\eqref{eq:Dirac and ASD Operators Relation} to construct zero modes for the Nahm Dirac operators associated to a Nahm solution $(T_0,T_1,T_2,T_3)$. 
Here
$$\mathfrak{D}_{\vec{x}}^\dagger: L^2(S^- \otimes E \otimes I) \to H^{-1}(S^+ \otimes E \otimes I),\quad \mathfrak{D}_{\vec{x}}:H^1(S^+ \otimes E \otimes I) \to L^2(S^- \otimes E \otimes I),$$ where\footnote{
One often sees $\mathfrak{D}_{\vec{x}}=\Delta$,
$\mathfrak{D}_{\vec{x}}^\dagger=\Delta\sp\dagger$ and $\Delta w=0$, $\Delta\sp\dagger v=0$ in the literature.}
\begin{align*} \label{eq:TwistedNahmDiracOperator}
\mathfrak{D}_{\vec{x}}^\dagger = i\frac{d}{ds}+iT_0-\sum_{j=1}^3 \sigma_j \otimes (T_j-ix_j), \  \mathfrak{D}_{\vec{x}}=i\frac{d}{ds}+iT_0+\sum_{j=1}^3 \sigma_j \otimes (T_j - ix_j).
\end{align*}
(To obtain the untwisted Dirac operators simply set $(x_1,x_2,x_3)$ to zero.)

After reduction of $\mathbb{R}^4$ to $\mathbb{R}$, the pair of commuting operators $(\mathcal{L},\mathcal{M})$ of \eqref{eq:commuting operators} is related to the Lax pair $(L,M)$ by $\mathcal{M}=i\left[ \frac{d}{ds}+M^N\right]$ and $\mathcal{L}=i L^N$. Indeed, in the one-dimensional case, the Lax pair is simply the more useful reformulation of $(\mathcal{L},\mathcal{M})$. 
Thus, a solution $\chi$ to $L^N \chi=0$ and $(\frac{d}{ds}+M^N) \chi=0$ is also a solution to $\mathcal{L}\chi=0$ and $\mathcal{M} \chi=0$.

\begin{proposition}[\cite{nahm_algebraic_1983}]
Let the spectral curve $\mathcal{C}$ be given by $\overset{n}{\underset{j=1}{\prod}} (\eta-p_j(\zeta))=0$. For $\vec{x}=(x_1,x_2,x_3)$, let $a_{jx}$ and $a_{xj}$ be the two roots of $p_j(\zeta)-p_x(\zeta)$, where $p_x(\zeta):=x_1+ix_2-2x_3\zeta-(x_1-ix_2)\zeta^2$. 

A $2n \times 2n$ fundamental matrix $W$ of solutions to
\begin{align*}
    \mathfrak{D}_{\vec{x}} w=0
\end{align*}
is given by the collection of spinors 
\begin{align*}
     e^{sh^+_x(a_{jx})}\begin{pmatrix} 1 \\ a_{jx} \end{pmatrix} \otimes U_j^N(s,a_{jx}),  \ \ \ \ e^{sh^+_x(a_{xj})}\begin{pmatrix} 1 \\ a_{xj} \end{pmatrix} \otimes U_j^N(s,a_{xj}),
\end{align*}
for $j=1,2,\dots,n$. Here, $U^N(s,\zeta)$ is an orthonormal basis of $H^0(S,L^s(n-1))$ in North coordinates and  $h^+_x(\zeta)=x_3+(x_1-ix_2)\zeta$.
\end{proposition}
\begin{proof}
The twist $\vec{x}$ shifts the Lax pair, $L_{\vec{x}}^N:=L^N-p_x(\zeta)$ and $M_{\vec{x}}^N:=M^N-h_x^+(\zeta)$. The shift in $M^N$ forces a prefactor of $e^{sh^+_x(\zeta)}$ for the basis $U^N$ and the eigenvalues of $L^N$ are shifted to $p_j(\zeta)-p_x(\zeta)$. Thus, we have
\begin{align*}
&\left[L_{\vec{x}}^N - \eta \right]  e^{sh^+_x(\zeta)} U_j^N(s,\zeta)=0, \\
&\left[ \frac{d}{ds}+M_{\vec{x}}^N \right] e^{sh^+_x(\zeta)} U_j^N(s,\zeta)=0.
\end{align*}
In order to apply \eqref{eq:Dirac and ASD Operators Relation}, we evaluate $\eta=p_j(\zeta)-p_x(\zeta)$ at its two roots $\zeta=a_{jx}$ and $\zeta=a_{xj}$.
\end{proof}

Let $W$ and $V$ be the fundamental matrices of $2n \times 2n$ solutions to $\mathfrak{D}_{\vec{x}} w=0$ and $\mathfrak{D}_{\vec{x}}^\dagger v=0$, respectively. We can choose $V$ to be $(W^\dagger)^{-1}$. To obtain the unique normalizable solution $v$, we must construct the projector of the fundamental matrix $V$ onto the $L\sp2$ kernel of $\mathfrak{D}_{\vec{x}}^\dagger$. This projector may be determined algebraically by cancellation of poles \cite[pp.4]{braden_construction_2018}.

\subsection{Monopole-sided Zero Modes}
We conclude by constructing the normalizable zero modes to the monopole Dirac operator $\mathfrak{D}_s^\dagger$ associated to the Dirac multi-monopole $(A,\Phi)$. We use an ansatz of Lamy-Poirier \cite{lamy-poirier_dirac_2015}, which builds on an observation of Nahm \cite{nahm_algebraic_1983} that we can go further with the solutions $\chi(\zeta)$ to $\mathcal{L} \chi(\zeta)=0$, $\mathcal{M} \chi(\zeta)=0$ to produce zero modes not only of $\mathfrak{D}$ but also of $\mathfrak{D}^\dagger$. 

First, the pair $(E,A)$ is anti-self-dual so that $\mathfrak{D}^\dagger \mathfrak{D}$ is a diagonal operator, i.e. $\mathfrak{D}^\dagger \mathfrak{D}=\mathds{1} \otimes \left(-\overset{4}{\underset{j=1}{\sum}} D_j^2\right)$. From \eqref{eq:Dirac and ASD Operators Relation}, it follows that $\mathfrak{D}^\dagger\left[ \mathfrak{D} \begin{pmatrix} 1 \\ \zeta \end{pmatrix} \otimes \chi(\zeta)\right]=0$. Since $\mathfrak{D}^\dagger \mathfrak{D}$ is diagonal, 
$$0=\mathfrak{D}^\dagger \mathfrak{D} \begin{pmatrix} 1 \\ \zeta \end{pmatrix} \otimes \chi(\zeta)=
\begin{pmatrix} 1 \\ \zeta \end{pmatrix} \otimes 
\left( -D_j^2 \chi(\zeta)\right),$$ therefore
$\chi(\zeta)$ is harmonic for all $\zeta$.
Moreover
\begin{align}
     \mathfrak{D} \begin{pmatrix} 1 \\ 0 \end{pmatrix} \otimes \chi(\zeta) .
\end{align}
solves the Dirac equation.
However, $\mathfrak{D} \begin{pmatrix} 1 \\ 0 \end{pmatrix} \otimes \chi(\zeta)$ might not be normalizable. To obtain a normalizable zero mode, one must do more. The Dirac operators are independent of $\zeta$, thus if $\chi$ is a zero mode then so is $F(\zeta)\chi$ for any function $F(\zeta)$. By linearity, $\underset{j}{\sum} F_j(\zeta) \chi_j(\zeta)$ is a zero mode. 

In the case of the Dirac multimonopole, Lamy-Poirier \cite{lamy-poirier_dirac_2015} shows that in order to ensure cancellation of poles and obtain a normalizable zero mode, one may specifically set $F_j(s,\zeta)={Q_j(s,\zeta)}/{\zeta}$ for a $n$-tuple $\left(Q_1(s,\zeta),Q_2(s,\zeta),\dots,Q_n(s,\zeta)\right)$ of polynomials satisfying the matching conditions $Q_i(a_{ij})=e^{-sr_{ij}}Q_j(a_{ij})$ and then take their residues.

We explicitly write the formulations. Now
$$\mathfrak{D}_s^\dagger: H^{-1}(S^- \otimes E \otimes L) \to L^2(S^+ \otimes E \otimes L),\quad \mathfrak{D}_s:L^2(S^+ \otimes E \otimes L) \to H^1(S^- \otimes E \otimes L),$$ where
\begin{align} \label{eq:TwistedDiracOperator}
\mathfrak{D}_s^\dagger = -\sum_{j=1}^3 \sigma_j \otimes D_j - i\Phi+s, \ \ \ \mathfrak{D}_s=\sum_{j=1}^3 \sigma_j \otimes D_j - i\Phi+s.
\end{align}
(To obtain the untwisted Dirac operators simply set $s$ to zero.)

For the multimonopole configuration of $n$ point monopoles of unit charge at distinct locations $\vec{c}_k, \ k=1,\dots,n$ in $\mathbb{R}^3$ the pair $(A,\Phi)$ can be written as \cite{cheng_fermion_2013}
\begin{align} \label{eq:MonopoleSolutions}
\Phi(x)=\sum_{1}^n \frac{i}{2r_k}, \ \ \  A(x)= \sum_{k=1}^n \frac{z_k d\bar{z_k}-\bar{z_k}dz_k}{4r_k(r_k+x_k)},
\end{align}
where for the vector $\vec{x}_k=\vec{x}-\vec{c}_k$ we set $r_k:=|\vec{x}_k|$, $z_k:=x_k^1+i x_k^2$, and $x_k:=x_k^3$. 
To avoid pathologies in the gauge, we rotate, if necessary, the monopole configuration in $\mathbb{R}^3$ so that no two monopoles are separated by a translation in the $x^3$ direction.

It is not difficult to calculate $\chi$ solving $\mathcal{L} \chi=0$ and $\mathcal{M} \chi=0$ since $(A,\Phi)$ is a superposition of $n$ point monopoles, which gives that $\chi$ is a product of its single point constituents. Let $a_{xk}$ be the root of $p_k(\zeta)-p_x(\zeta)$ corresponding to the direction from $\vec{x}$ to the monopole located at position $\vec{a}_k$. One obtains $\chi=\overset{n}{\underset{k=1}{\prod}}\sqrt{\frac{-a_{xk}}{\bar{z}_k}} \frac{1}{\zeta-a_{xk}}$. The twist by $s$ gives an additional prefactor of $e^{-sh_{xk}^-(\zeta)}$ to each constituent.  

\begin{proposition}[{\cite[Eqn (5.6)]{lamy-poirier_dirac_2015}}]
Let $\left(Q_1(\zeta),Q_2(\zeta,\dots,Q_n(\zeta)\right)$ be a tuple of polynomials satisfying the matching conditions $Q_i(a_{ij})=e^{-sr_{ij}}Q_j(a_{ij})$ constructed above. Let $a_{xj}$ be the root of $p_j(\zeta)-p_x(\zeta)$ that points from $\vec{x}$ to $\vec{a}_j$ for $p_x(\zeta):=x_1+ix_2-2x_3\zeta-(x_1-ix_2)\zeta^2$. Then
\begin{align*}
    \mathfrak{D}_s \begin{pmatrix}1 \\ 0 \end{pmatrix} \overset{n}{\underset{i=1}{\sum}}  \frac{1}{2\pi i} \oint_{a_{xi}} 
  \frac{Q_i(\zeta)}{\zeta} e^{-s(x_i-\frac{z_i}{\zeta})} \left(\overset{n}{\underset{k=1}{\prod}} \sqrt{\frac{-a_{xk}}{\bar{z}_k}} \frac{1}{\zeta-a_{xk}}\right)\,  d\zeta
\end{align*}
is a normalizable zero mode of $\mathfrak{D}^\dagger_s$.
\end{proposition}

\bibliographystyle{alpha}
 \bibliography{./main}
\end{document}